\documentclass{ecai} 

\usepackage{latexsym}
\usepackage{amssymb}
\usepackage{amsmath}
\usepackage{amsthm}
\usepackage{booktabs}
\usepackage{enumitem}
\usepackage{graphicx}
\usepackage{color}

\usepackage{xspace}
\setitemize{noitemsep,topsep=2pt,parsep=3pt,partopsep=2pt,label = --}

\input{macros.sty}

\newtheorem{theorem}{Theorem}
\newtheorem{lemma}[theorem]{Lemma}
\newtheorem{corollary}[theorem]{Corollary}
\newtheorem{proposition}[theorem]{Proposition}

\newcommand{\BibTeX}{B\kern-.05em{\sc i\kern-.025em b}\kern-.08em\TeX}

\usepackage{tikz}
\usetikzlibrary{automata}
\usetikzlibrary{arrows}
\usetikzlibrary{backgrounds}
\usetikzlibrary{positioning}
\usetikzlibrary{fit}
\usetikzlibrary{calc}
\usetikzlibrary{matrix}
\usetikzlibrary{patterns}
\usetikzlibrary{shapes.geometric}
\usetikzlibrary{decorations.pathreplacing}
\usetikzlibrary{decorations.pathmorphing}

\begin{document}


\begin{frontmatter}


\paperid{616} 


\title{Computational Complexity of Standpoint LTL}


\author[A]{\fnms{St\'ephane}~\snm{Demri}}
\author[B]{\fnms{Przemys{\l}aw Andrzej}~\snm{Wałęga}
}

\address[A]{Universit{\'e} Paris-Saclay, ENS Paris-Saclay, CNRS, LMF, 91190, Gif-sur-Yvette, France}
\address[B]{University of Oxford, Queen Mary University of London, United Kingdom}


\begin{abstract}
  Standpoint linear temporal logic \SLTL is a recent formalism able
  to model possibly conflicting commitments made by distinct agents, taking into account
  aspects of temporal reasoning. 
  In this paper, we analyse  the computational properties of  $\SLTL$.
  First,  we establish logarithmic-space reductions between the satisfiability problems for
  the multi-dimensional modal logic
  $\PTLSfive$
  and 
  $\SLTL$. This leads to the \expspace-completeness of the satisfiability problem in $\SLTL$,
  which is a surprising result in view of previous investigations.
  Next, we present a method of restricting $\SLTL$ so that the obtained fragment is a strict extension of both
  the (non-temporal) standpoint logic and
    linear-time temporal logic $\LTL$, but the  satisfiability problem
  is \pspace-complete in this fragment. 
  Thus, we show how to combine standpoint logic with $\LTL$ so that the worst-case
  complexity of the obtained combination is not higher than of pure $\LTL$. 
\end{abstract}

\end{frontmatter}


\section{Introduction}
\label{section-introduction}

\paragraph{Standpoint Logic.} 
Recently, a new framework based on modal logics was developed in order
to interpret languages in the presence of vagueness~\citep{GomezAlvarez19}. 
The framework is called `standpoint logic'
where standpoints (a concept first introduced by~\citet{Bennett06}
in a logical context, see also~\cite{Bennett11})
are used to interpret vague expressions.
Logical reasoning about vagueness has a long tradition stemming from fuzzy logics~\cite{Zadeh65,Zadeh75},
to information logics based on
rough sets~\cite{Pawlak82,Orlowska94,Banerjee&Chakraborty&Szalas24}.
In standpoint logic, each standpoint $\astandpoint$ is associated  with modalities $\Diamond_{\astandpoint}$ and
$\Box_{\astandpoint}$,  and with a set of interpretations (a.k.a. {\em precisifications})
corresponding to $\astandpoint$. 
While $\Diamond_{\astandpoint} \aformula$ reads as ``according to $\astandpoint$, it is conceivable that
$\aformula$'', dually, $\Box_{\astandpoint} \aformula$ reads as ``according to $\astandpoint$, it is unequivocal that $\aformula$''. 
The language of standpoint logic is equipped also with a binary operator $\preceq$ between standpoints, such that
$\astandpoint \preceq \astandpoint'$ is interpreted as ``the standpoint $\astandpoint$
is sharper than $\astandpoint'$'' leading to the validity of the modal axiom
$\Box_{\astandpoint'}  \avarprop \implication \Box_{\astandpoint}  \avarprop$.
This is reminiscent to partially-ordered ($\Sfour$) modal operators~\citep{Allwein&Harrison10},
role hierarchies in description logics~\citep{Baaderetal17},
and more generally to grammar logics~\citep{Farinas&Penttonen88},
for which computational properties are well studied~\citep{Baldoni&Giordano&Martelli98,Demri01}. 
Originally, standpoint logic framework was developed
for propositional calculus~\citep{GomezAlvarez19}, then also for predicate logic and
description logics~\citep{GomezAlvarez&Rudolph&Strass23,GomezAlvarezetal23}, and most recently for the
temporal logic $\LTL$~\citep{Gigante&GomezAlvarez&Lyon23}.
This framework has a potential for further combinations with logical formalisms
dedicated to knowledge representation and reasoning, allowing us to perform logical reasoning
about vagueness~\citep{Fine75,Bennett98,Bennett11}.

\paragraph{Standpoint \LTL.} As evoked above,
the recent paper~\citep{Gigante&GomezAlvarez&Lyon23}
introduced a multi-perspective approach by combining standpoints and temporal reasoning
expressed in the linear-time temporal logic $\LTL$~\cite{Pnueli77}, which is a very popular specification
language, for instance used for model-checking~\cite{baier2008principles},
temporal planning~\cite{Calvanese&DeGiacomo&Vardi02,Aminofetal19}, and temporal
reasoning with description logics~\cite{Baader&Ghilardi&Lutz12}.
This new formalism, called $\SLTL$, handles both evolutions of systems and changes of standpoints. 
\citet{Gigante&GomezAlvarez&Lyon23} develop a tableau-based proof system to reason about $\SLTL$, leading to a computational
analysis for deciding the satisfiability status of $\SLTL$ formulae.
Within $\SLTL$, each standpoint is interpreted as a set of
\LTL models (i.e. as a set of $\omega$-sequences of
propositional valuations, also known as traces) and the logical formalism has the ability to model possibly conflicting
commitments made by distinct agents. Hence,  $\SLTL$ significantly increases the
modelling capabilities offered by $\LTL$, and so, $\SLTL$ can be seen as a non-trivial extension of $\LTL$.
Our initial motivation in this work is to understand the computational properties of
$\SLTL$.
We agree with~\citet{Gigante&GomezAlvarez&Lyon23} that it is particularly desirable  to be able to decide
the satisfiability status of $\SLTL$ formulae in polynomial-space, which is
the best we can hope for in view of \pspace-completeness of \LTL~\cite{Sistla&Clarke85}.
However, in contrast to the results of~\citet{Gigante&GomezAlvarez&Lyon23}, we show that the \pspace membership can be guaranteed only for strict fragments of $\SLTL$
(see Section~\ref{section-pspace}) and \expspace is required for the full  $\SLTL$ (see Section~\ref{section-expspace-completeness}).

\paragraph{Our contributions.} We study the computational properties of
the satisfiability problem for the standpoint linear temporal logic $\SLTL$. 
We show that the problem is \expspace-complete (Theorem~\ref{theorem-expspace-completeness})
by establishing logarithmic-space reductions between $\SLTL$ and the multi-dimensional modal logic
$\PTLSfive$ whose satisfiability problem is known to be \expspace-complete~\citep{Gabbayetal03}.  
The obtained  \expspace-completeness of $\SLTL$ contrasts
with \pspace-membership   claimed by~\citet{Gigante&GomezAlvarez&Lyon23}.
In Section~\ref{section-example-expspace}, we  provide 
examples of \SLTL{} formulae whose satisfiability is challenging to check as their models
have an infinite set of traces and at every position, an exponential amount of valuations
are witnessed on such traces (see Proposition~\ref{proposition-expspace-examples}).
\cut{
In Section~\ref{section-pspace}, we  identify two fragments of $\SLTL$ which 
contain both \LTL and propositional standpoint logic, but the satisfiability problem can be decided in them
in polynomial-space.
}
In Section~\ref{section-no-ltl-in-boxs}, we  identify a fragment of $\SLTL$ which 
contains both \LTL and propositional standpoint logic, but the satisfiability problem can be decided 
in polynomial-space.
Since  $\LTL$ is known to be \pspace-complete~\cite{Sistla&Clarke85}, we obtain the same tight complexity result for the newly introduced fragment of \SLTL (Theorem~\ref{theorem-ltlpsl}).
To do so, we use the
automata-based approach following the general principles for \LTL~\citep{Vardi&Wolper94},
but with
non-trivial modifications.

\section{Logical Preliminaries}
\label{section-preliminaries}

In this section, we briefly introduce the standpoint linear temporal logic \SLTL,
the propositional standpoint logic \PSL, 
as well as the multi-dimensional modal logic $\PTLSfive$,
which we exploit in Section~\ref{section-expspace-completeness}.
For motivations and detailed presentation of these  logics,  we refer a reader
to~\citep{alvarez2022standpoint,Gigante&GomezAlvarez&Lyon23} and~\citep[Chapter 5]{Gabbayetal03}.

\subsection{Standpoint linear temporal logic $\SLTL$}
\label{section-introduction-sltl}

The $\SLTL$ \defstyle{formulae} are built over a countably infinite set $\varprop$ of \defstyle{propositional variables} and a countably infinite set $\standpoints = 
\set{\astandpoint, \astandpoint',\ldots}$ of \defstyle{standpoint symbols}
including the \defstyle{universal standpoint symbol} $\universalstandpoint$.
The $\SLTL$ formulae $\aformula$ are defined
according to the grammar:
\cut{
\begin{tightcenter}
$
\aformula ::=
\avarprop \mid
\astandpoint \preceq \astandpoint' \mid
\neg \aformula \mid
\aformula \wedge \aformula \mid
\Diamond_{\astandpoint} \aformula \mid \Box_{\astandpoint} \aformula \mid 
\mynext \aformula \mid \aformula \until \aformula
$,
\end{tightcenter}
}
\[
\aformula ::=
\avarprop \mid
\astandpoint \preceq \astandpoint' \mid
\neg \aformula \mid
\aformula \wedge \aformula \mid
\Diamond_{\astandpoint} \aformula \mid \Box_{\astandpoint} \aformula \mid 
\mynext \aformula \mid \aformula \until \aformula,
\]
where $\avarprop \in \varprop$  and $\astandpoint,\astandpoint' \in \standpoints$;
other  Boolean connectives and \LTL temporal operators (e.g. $\implication$, $\leftrightarrow$,
$\vee$, $\always$ for ``always in the future'', and $\sometimes$ for ``sometime in the future'') are treated as standard abbreviations. In terms  of expressivity, one modality among $\Diamond_{\astandpoint},\Box_{\astandpoint}$
is sufficient. 
An $\SLTL$ \defstyle{model}, $\amodel$, is a structure of the form $\amodel = \pair{\Pi}{\lambda}$ where,
\begin{itemize}
\itemsep 0 cm 
\item $\Pi \neq \emptyset$ is a 
set of $\LTL$ models (traces) of the form
$\altlmodel: \Nat \longrightarrow \powerset{\varprop}$,
\item $\lambda$ is a map of the form $\lambda: \standpoints \longrightarrow (\powerset{\Pi} \setminus \set{\emptyset})$
  such that $\lambda(\universalstandpoint) = \Pi$. 
\end{itemize}
For example, Figure~\ref{figure-m-construction} 
presents
an \SLTL model with six traces. 
The \defstyle{satisfaction relation}, $\amodel, \altlmodel, i \models \aformula$, for an \SLTL model 
$\amodel = \pair{\Pi}{\lambda}$, $\sigma \in \Pi$, 
$i \in \Nat$, and an \SLTL formula $\varphi$  is 
defined inductively as follows (we omit the standard clauses for Boolean connectives):
\begin{align*}
& \amodel, \altlmodel, i \models \avarprop 
&& \equivdef && \avarprop \in \altlmodel(i),
\\
& \amodel, \altlmodel, i \models \astandpoint \preceq \astandpoint' && \equivdef &&
\lambda(\astandpoint) \subseteq \lambda(\astandpoint'),
\\
& \amodel, \altlmodel, i \models \Diamond_{\astandpoint} \aformula && \equivdef &&
\amodel, \altlmodel', i \models \aformula, \text{ for some } \altlmodel' \in \lambda(\astandpoint),
\\
& \amodel, \altlmodel, i \models \Box_{\astandpoint} \aformula && \equivdef &&
\amodel, \altlmodel', i \models  \aformula,
\text{ for all } \altlmodel' \in \lambda(\astandpoint)
,
\\
& \amodel, \altlmodel, i \models \mynext \aformula && \equivdef && \amodel, \altlmodel, i+1 \models \aformula,
\\
& \amodel, \altlmodel, i \models \aformula \until \aformulabis && \equivdef &&
\text{there is } i' \geq i \text{ such that } \amodel, \altlmodel, i' \models \aformulabis 
\\
& && &&  \text{and } \amodel, \altlmodel, i'' \models \aformula \text{ for all } i \leq  i'' < i'.
\end{align*}
The satisfiability problem for $\SLTL$ takes as input an $\SLTL$ formula $\aformula$ and asks whether
there is a model $\amodel = \pair{\Pi}{\lambda}$ and  $\altlmodel \in \Pi$ such that
$\amodel, \altlmodel, 0 \models \aformula$. 
By way of example, we provide an \SLTL formula below taken from the medical devices example~\cite[Section 2.1]{Gigante&GomezAlvarez&Lyon23}
(more elaborated examples can be found
in~\citep{GomezAlvarez19,alvarez2022standpoint,Gigante&GomezAlvarez&Lyon23}):
\cut{
\begin{tightcenter}
  $
  \Box_{\universalstandpoint} (   \always \neg \mathit{Malf}  \to \mathit{Test} ) \land 
  \Box_{\mathsf{IT}} (  \mathit{Comp} \vee \mathit{Test} \to \mathit{Safe} ).
  $
\end{tightcenter}
}
\[
  \Box_{\universalstandpoint} (   \always \neg \mathit{Malf}  \to \mathit{Test} ) \land 
  \Box_{\mathsf{IT}} (  \mathit{Comp} \vee \mathit{Test} \to \mathit{Safe} ).
\]
The first conjunct states that all countries agree in their standpoints that if a
medical device never malfunctions, then it is safe according to testing. 
The second one states that Italy deems a device safe if it is safe according to testing
or it has been found safe by comparison.

Regarding our definition of \SLTL, it is worth observing that, as far as we can judge, the satisfiability problem is not formally
defined in~\cite[Section 2.2]{Gigante&GomezAlvarez&Lyon23}. Only
  the validity problem is defined. 
  In particular, non validity of the $\SLTL$ formula
  $\aformula$ is defined 
  as the existence of
  an \SLTL model $\amodel = \pair{\Pi}{\lambda}$ and $\altlmodel \in \Pi$ such that
  $\amodel, \altlmodel, 0 \models \neg \aformula$.
  So, our  definition of satisfiability is  dual to the notion of
  validity used by~\citet{Gigante&GomezAlvarez&Lyon23}.

Besides,  it is worth noting that the above presentation of $\SLTL$ differs slightly with the
definition of~\citet[Section 2.2]{Gigante&GomezAlvarez&Lyon23}, but it has no impact on our results, as described next.
$\SLTL$ formulae, as defined in~\cite{Gigante&GomezAlvarez&Lyon23},
are in negation normal form and allow for using the ``release'' 
\LTL operator.
In our definition negation is unrestricted and we do not use the  ``release'' 
operator, which
makes no substantial difference.
\citet{Gigante&GomezAlvarez&Lyon23} assumes also that the formulae $\astandpoint \preceq \astandpoint'$
  cannot be combined with other formulae, so
the way we define $\SLTL$ formulae  is slightly more expressive.
However, our \expspace-hardness proof (reduction in Lemma~\ref{reducetoSLTL})  does not use  formulae of the form $\astandpoint \preceq \astandpoint'$
(actually only the modalities $\Box_{\universalstandpoint}$, $\mynext$, and $\until$ are needed).
The \expspace-membership (Corollary~\ref{inEXPS}) for our, slightly richer language, clearly implies the same upper bound for the weaker
language of \citet{Gigante&GomezAlvarez&Lyon23}.

The change of standpoint performed with the modalities
$\Diamond_{\astandpoint}$ and $\Box_{\astandpoint}$ is reminiscent to the change of observational power studied
in~\cite{Barriereetal19} with the modalities $\Delta^o$. In both cases, a modality explicitly performs
a change in the way the forthcoming formulae are evaluated. 

\subsection{Propositional standpoint logic $\PSL$}
\label{section-introduction-psl}

In the sequel, we also consider \defstyle{propositional standpoint logic}~\citep{alvarez2022standpoint}
(herein, written $\PSL$) understood as the fragment of $\SLTL$ without
temporal connectives. The grammar of formulae is restricted
to
\[
\aformula ::=
\avarprop \mid
\astandpoint \preceq \astandpoint' \mid
\neg \aformula \mid
\aformula \wedge \aformula \mid
\Diamond_{\astandpoint} \aformula \mid \Box_{\astandpoint} \aformula,
\]
and the models are of the form $\amodel = \pair{\Pi}{V}$ where $\Pi$ is a finite non-empty set of  \defstyle{precisifications},
$V: \standpoints \cup \varprop \longrightarrow \powerset{\precisis}$ is a \defstyle{valuation}
such that
for all $\astandpoint \in \standpoints$, we have $V(\astandpoint) \neq \emptyset$
and $V(\universalstandpoint) = \precisis$.
The satisfaction relation is defined as follows (where $\aprecisi \in \Pi$ and we omit the obvious clauses for Boolean connectives):
\begin{align*}
  &\amodel, \aprecisi \models \avarprop && \equivdef &&
  \aprecisi \in V(\avarprop),
\\
& \amodel, \aprecisi \models \astandpoint \preceq \astandpoint' && \equivdef && V(\astandpoint) \subseteq V(\astandpoint'),
\\
& \amodel, \aprecisi \models \Diamond_{\astandpoint} \aformula && \equivdef &&
\amodel, \aprecisi' \models \aformula, \text{ for some } \aprecisi' \in V(\astandpoint),
\\
& \amodel, \aprecisi \models \Box_{\astandpoint} \aformula && \equivdef &&
\amodel, \aprecisi' \models \aformula, \text{ for all } \aprecisi' \in V(\astandpoint). 
\end{align*}
The satisfiability problem, for  an input formula $\aformula$, consists in checking whether there is some pair $\amodel, \aprecisi$  such that
$\amodel, \aprecisi \models \aformula$. This problem
is \np-complete.
To show \np-membership, ~\citet[Section 4.4.2]{GomezAlvarez19} proved that if a satisfiable formula $\aformula$ contains
$N_1$ many standpoint symbols and $N_2$ many diamond modal operators, then $\aformula$ is satisfied in a  model
$\amodel = \pair{\precisis}{V}$ such that $\card{\Pi} \leq N_1 + N_2 +1$. 
An alternative way to get
the \np-membership is to translate
$\aformula$ into a  formula $\bigwedge_{\astandpoint} \Diamond \astandpoint \wedge  \atranslation(\aformula)$
of the modal logic $\Sfive$~\cite{Blackburn&deRijke&Venema01}, 
where $\atranslation$ turns standpoint operators into  modal operators in the following manner: 
$\atranslation(\astandpoint \preceq \astandpoint') = \Box(\astandpoint \implication \astandpoint')$,
$\atranslation(\Diamond_{\astandpoint} \aformula)
= \Diamond(\astandpoint \wedge \atranslation(\aformula))$, and
$\atranslation(\Box_{\astandpoint} \aformula) = \Box (  \astandpoint \to \atranslation(\aformula))$.
The correctness of such a reduction relies naturally on the Kripke-style semantics for $\PSL$.
We will refine complexity analysis of \PSL in Section~\ref{section-no-ltl-in-boxs}, which will be essential to establish complexity of \SLTL fragments.

\subsection{Multi-dimensional modal logic $\PTLSfive$}

Another logic that is useful herein is the multi-dimensional modal logic
$\PTLSfive$~\citep[Chapter 5]{Gabbayetal03} defined as the product of $\LTL$ and $\Sfive$.
$\PTLSfive$ formulae are generated from the grammar
\[
\aformula::=
\avarprop \mid
\neg \aformula \mid 
\aformula \wedge \aformula \mid
\Diamond \aformula \mid \Box \aformula \mid 
\mynext \aformula \mid \aformula \until \aformula,
\]
where $\avarprop \in \varprop$ is a propositional
variable. As in \SLTL, we use standard abbreviations for other Boolean connectives and \LTL operators
($\implication$, $\vee$, $\always$, $\sometimes$, etc.).
The \defstyle{models} for $\PTLSfive$ are of the form $\amodel = \triple{\Nat \times W}{R}{L}$ where
$\pair{W}{R}$ is an $\Sfive$-frame (i.e. $R$ is an equivalence relation on $W$)
and $L: \Nat \times W \longrightarrow \powerset{\varprop}$. The satisfaction relation for
$\PTLSfive$ is defined as follows (again, we omit the standard clauses for Boolean connectives):
\begin{align*}
& \amodel, \pair{n}{w}  \models \avarprop  
&& \hspace{-1em}\equivdef && \hspace{-1em} \avarprop \in L(n,w),
\\
& \amodel, \pair{n}{w} \models \Diamond \aformula && \hspace{-1em}\equivdef && \hspace{-1em}
\amodel, \pair{n}{w'} \models \aformula,
  \text{ for some } w' \in R(w),
\\  
& \amodel, \pair{n}{w} \models \Box \aformula && \hspace{-1em}\equivdef && \hspace{-1em}
\amodel, \pair{n}{w'} \models \aformula,
  \text{ for all } w' \in R(w),
\\
& \amodel, \pair{n}{w} \models \mynext \aformula && \hspace{-1em}\equivdef && \hspace{-1em}
\amodel, \pair{n+1}{w} \models \mynext \aformula,
\\
& \amodel, \pair{n}{w} \models \aformula \until \aformulabis && \hspace{-1em}\equivdef && \hspace{-1em}
  \text{there is } n' \geq n \text{ such that } \amodel, \pair{n'}{w} \models \aformulabis 
\\
&  &&  && \hspace{-1em}  \text{and } \amodel, \pair{n''}{w} \models \aformula \text{ for all } n \leq n'' < n'.
\end{align*}
Therefore, the modalities $\Diamond$  and $\Box$ allow us to move within
the $\pair{W}{R}$ dimension  whereas the temporal connectives $\mynext$ and $\until$
allow us to move along the $\pair{\Nat}{\leq}$ dimension. 
The satisfiability problem for $\PTLSfive$ takes as input a $\PTLSfive$ formula $\aformula$ and asks whether
there is a model $\amodel = \triple{\Nat \times W}{R}{L}$ and $\pair{n}{w} \in \Nat \times W$
such that $\amodel, \pair{n}{w} \models \aformula$.
It is  known, that satisfaction of a formula can be always witnessed by a model $\amodel = \triple{\Nat \times W}{R}{L}$ with $R= W \times W$ and by a pair $(0,w)$ (i.e. its first component is the origin position $0$).
In the sequel we will use this assumption; in particular, we will assume that $R= W \times W$, and for simplicity of presentation we will drop the component $R$ from $\PTLSfive$ models.

It is worth noting that the above presentation of $\PTLSfive$ differs slightly from the
definitions by~\citet[Section 2.1]{Gabbayetal03}, but it has no impact on our results.
Indeed, \citet{Gabbayetal03}
use only the strict ``until'' operator, which we denote by $\until_{<}$ (and no next-time operator $\mynext$) and whose semantics is as follows:
\begin{align*}
& \amodel, \pair{n}{w} \models \aformula \until_{<} \aformulabis &&\hspace{-1.4em}\equivdef && \hspace{-1.1em}
  \text{there is } n' > n \text{ with } \amodel, \pair{n'}{w} \models \aformulabis 
\\
& && && \hspace{-2em} \text{and } \amodel, \pair{n''}{w} \models \aformula \text{ for all } n < n'' < n'.
\end{align*}
It is easy to see that  $\aformula \until_{<} \aformulabis$ can be encoded by 
$\mynext(\aformula  \until  \aformulabis)$
and therefore the \expspace-hardness for $\PTLSfive$ proved by~\citet[Theorem 5.43]{Gabbayetal03} applies
also to our version of $\PTLSfive$. 
As far as the \expspace-membership is concerned, the satisfiability
problem for our version of $\PTLSfive$ is in \expspace using the approach of~\citet[Theorem 6.65]{Gabbayetal03} dedicated
to $\PTLSfive$ with strict ``until'' and using a standard renaming technique~\citep[Proposition 2.10]{Gabbayetal03}.
Note that a naive translation $\atranslation$ 
from our language to a formula with strict ``until'' exploiting
 $\atranslation(\aformula \until \aformulabis) = \atranslation(\aformulabis) \vee 
(\atranslation(\aformula) \wedge
\atranslation(\aformula) \until_{<}
\atranslation(\aformulabis) )$,
would 
cause an exponential blow-up.
However, we can get a logarithmic-space reduction using the renaming technique,
where any subformula $\aformulater$ is associated with a fresh propositional variable $\avarprop_{\aformulater}$.
For instance, to capture the meaning of $\aformula \until \aformulabis$ we introduce an additional formula
$
\Box \always
\big(
\avarprop_{\aformula \until \aformulabis} \leftrightarrow
\avarprop_{\aformulabis} \vee 
( \avarprop_{\aformula} \wedge 
\avarprop_{\aformula} \until_{<}
\avarprop_{\aformulabis})\big).
$
This additional formula (propagating an equivalence all over the model),
if asserted in any  world of the form $\pair{0}{w}$, allows us to state that
$\avarprop_{\aformula \until \aformulabis}$ is equivalent with $\aformulabis \vee 
( \aformula \wedge 
\aformula \until_{<}
\aformulabis)$, in all elements in $\Nat \times W$.
The propagation is over the model because of the modality $\Box \always$ and we can always
assume that a world satisfying our formula is of the form $\pair{0}{w}$. 
As a conclusion, the version of $\PTLSfive$ involved in this paper admits also
an \expspace-complete satisfiability problem
(\expspace-hardness follows from the fact that $\aformula \until_{<} \aformulabis$ can be encoded by 
$\mynext(\aformula  \until  \aformulabis)$).

Let us conclude this section by evoking the relationships between \SLTL and the well-known modal logic
$\Sfive$~\citep{Blackburn&deRijke&Venema01}.
The logic $\SLTL$ contains a modality $\Box_{\universalstandpoint}$ where $\universalstandpoint$
can be understood as the universal
standpoint interpreted by the total set of
traces and therefore  $\Box_{*}$ behaves
naturally as an $\Sfive$ modality, whence the component $\Sfive$ in $\PTLSfive$.
The presence of $\Sfive$ is not our finding as it has been already observed that
quantification over precisifications leads to $\Sfive$ modalities, see e.g., the works of~\citet[Section 2.1]{Bennett98}
and~\citet[page 43]{Bennett06}, as well as the presence of $\Sfive$ modalities for modelling
standpoints in~\cite[Section 3.5]{GomezAlvarez19} and in~\cite[Chapter 4]{GomezAlvarez19}.
Furthermore, the relationship with multi-dimensional modal logics is already
briefly evoked in~\cite[Section 7.3.3]{GomezAlvarez19}. More importantly,
an early introduction of some multi-dimensional modal logic is performed by~\citet[Section 2.1]{Bennett98} where
first-order logic and propositional modal logic $\Sfive$
are mixed. Hence, the fact that we use $\PTLSfive$ is not a total surprise, and the
need for multi-dimensional modal logics was in the air for some time.
Our contribution consists in establishing a formal result involving
multi-dimensional modal logics and in designing simple logarithmic-space reductions between
$\SLTL$ and  $\PTLSfive$, proving \expspace-hardness of $\SLTL$, despite the complexity
result claimed in~\cite[Theorem 28]{Gigante&GomezAlvarez&Lyon23}. 
The design of
a \pspace fragment completes our
analysis
and provides us with a fragment of $\SLTL$ which is a strict extension of both $\LTL$ and standpoint logic, but its
computational complexity is not higher than the complexity of pure $\LTL$.

\section{Satisfiability for $\SLTL$ is \expspace-complete}
\label{section-expspace-completeness} 

In this section we design logarithmic-space translations from $\SLTL$  to  $\PTLSfive$ formulae, and
vice versa.
As a result, we will obtain that the computational complexity of  satisfiability for $\SLTL$ formulae is the same as for   $\PTLSfive$ formulae, namely \expspace-complete.
Both of our translations  exploit similarities between
$\SLTL$ and  $\PTLSfive$ models.
More specifically,
we will consider an element $w\in W$ in an $\PTLSfive$ model $\pair{\Nat \times W}{L}$
as a name for an $\LTL$ trace $\sigma \in \Pi$ from an $\SLTL$ model $\pair{\Pi}{\lambda}$. 

\subsection{
Translation from $\PTLSfive$ to $\SLTL$}
\label{section-expspace-hardness}

Let $\aformula$ be a $\PTLSfive$  formula and $\atranslation_1(\aformula)$ be its translation
obtained from $\aformula$ by replacing every occurrence of $\Diamond$ by $\Diamond_{\universalstandpoint}$
and every occurrence 
$\Box$ by $\Box_{\universalstandpoint}$ (an alternative translation consists in using
$\Diamond_{\astandpoint}$ and 
and $\Box_{\astandpoint}$ for some fixed $\astandpoint \in \standpoints$).
We show that this simple translation preserves satisfiability.

\begin{lemma}\label{reducetoSLTL}
$\aformula$ is  $\PTLSfive$-satisfiable iff $\atranslation_1(\aformula)$ is  $\SLTL$-satisfiable.
\end{lemma}

\begin{proof}[Proof sketch] 
Assume that there is a $\PTLSfive$ model
  $\amodel = \pair{\Nat \times W}{L}$  and $\pair{0}{w} \in \Nat \times W$ such that
  $\amodel, \pair{0}{w} \models \aformula$. 
  Let us build an $\SLTL$ model $\amodel' = \pair{\Pi}{\lambda}$ and a trace $\altlmodel \in \Pi$
  such that
  $\amodel, \altlmodel,0 \models \atranslation_1(\aformula)$.
  To this end, we let $\Pi$ consist of all traces of the form
  \[
  \altlmodel_{w'} \egdef L((0,w')), L((1,w')), L((2,w')), L((3,w')), \ldots
  \]
  with $w' \in W$;
  note that we represent  a trace $\altlmodel': \Nat \longrightarrow \powerset{\varprop}$ as an $\omega$-sequence $\altlmodel'(0),\altlmodel'(1), \altlmodel'(2), \ldots$.
  We let $\lambda(\astandpoint) \egdef \Pi$, for each $\astandpoint \in \standpoints$ (for the construction it is only important that $\lambda(\universalstandpoint) = \Pi$ as no other standpoint symbol occurs in $\atranslation_1(\aformula)$).  
We can prove by 
structural induction that, for all subformulae $\aformulabis$ of $\aformula$ and
  for all $\pair{n}{w'} \in \Nat \times W$, we have
    $\amodel, \pair{n}{w'} \models \aformulabis$ iff 
    $\amodel', \altlmodel_{w'}, n \models \atranslation_1(\aformulabis)$.  
  By way of example
  we handle below the case when
  the subformula is of the form $\Box \aformulabis$.
  It suffices to observe that the following statements are equivalent:
  \begin{itemize}
  \itemsep 0 cm 
  \item $\amodel, \pair{n}{w'} \models \Box \aformulabis$
  \item $\amodel, \pair{n}{w''} \models \aformulabis$ for all
    $w'' \in W$  \hfill 
     (by definition of $\models$)
  \item $\amodel', \altlmodel_{w''}, n \models
        \atranslation_1(\aformulabis)$ for all
        $w'' \in W$
        \hfill (by  induction hypothesis)
  \item $\amodel', \altlmodel', n \models
        \atranslation_1(\aformulabis)$ for all
        $\altlmodel' \in \Pi$
        \hfill (by definition of $\Pi$)
  \item $\amodel', \altlmodel_{w'}, n \models
        \Box_{*} \atranslation_1(\aformulabis)$
        \hfill (by definition of $\models$)
  \item $\amodel', \altlmodel_{w'}, n \models
        \atranslation_1(\Box \aformulabis)$
        \hfill (by definition of $\atranslation_1$)
  \end{itemize}
Therefore, we can show that $\amodel, \altlmodel_w,0 \models \atranslation_1(\aformula)$.
    
  For the opposite implication, assume that 
  there are an $\SLTL$ model $\amodel = \pair{\Pi}{\lambda}$ and a trace
  $\altlmodel \in \Pi$
  such that
  $\amodel, \altlmodel,0 \models \atranslation_1(\aformula)$.
  Let us build a $\PTLSfive$ model  $\amodel' = \pair{\Nat \times W}{L}$ and $\pair{0}{w} \in \Nat \times W$ such that
  $\amodel', \pair{0}{w} \models \aformula$.
  We let $W \egdef \Pi$  and $L(\pair{n}{\altlmodel'}) \egdef \altlmodel'(n)$, for all $\pair{n}{\altlmodel'} \in \Nat \times W$.  
  We can show by structural induction that  for all subformulae $\aformulabis$ of $\aformula$,
   all $\altlmodel' \in \Pi$, and  all $n \in \Nat$, we have
    $\amodel, \altlmodel', n  \models \atranslation_1(\aformulabis)$  iff 
    $\amodel', \pair{n}{\altlmodel'} \models \aformulabis$.  
  By way of example, 
  we handle below the case when
  the subformula is of the form $\aformulabis_1 \until \aformulabis_2$. To this end, we observe that the below statements are equivalent: 
  \begin{itemize}
  \item $\amodel, \altlmodel', n  \models
    \atranslation_1(\aformulabis_1 \until \aformulabis_2)$
  \item $\amodel, \altlmodel', n  \models
    \atranslation_1(\aformulabis_1) \until \atranslation_1(\aformulabis_2)$
    \hfill (by definition of $\atranslation_1$)
  \item There is $n' \geq n$ such that
    $\amodel, \altlmodel', n'  \models \atranslation_1(\aformulabis_2)$
    and for all $n \leq n'' < n'$, we have
    $\amodel, \altlmodel', n''  \models \atranslation_1(\aformulabis_1)$
    \hfill (by definition of $\models$)
  \item There is $n' \geq n$ such that
    $\amodel', \pair{n'}{\altlmodel'} \models \aformulabis_2$
    and for all $n \leq n'' < n'$, we have
    $\amodel', \pair{n''}{\altlmodel'} \models \aformulabis_1$
    \hfill (by  induction hypothesis)
  \item $\amodel', \pair{n}{\altlmodel'} \models
    \aformulabis_1 \until \aformulabis_2$
    \hfill (by definition of $\models$)
  \end{itemize}     
   This allows us to show that $\amodel', \pair{0}{\altlmodel} \models \aformula$.
\end{proof}

\noindent As $\atranslation_1(\aformula)$ is computed in logarithmic-space, we get the following. 

\begin{corollary}\label{EXPShard}
$\SLTL$-satisfiability is \expspace-hard.
\end{corollary}
Note that this result
contradicts the \pspace upper bound shown by~\citet[Theorem 28]{Gigante&GomezAlvarez&Lyon23}, as \pspace and \expspace
are known to be distinct complexity classes. 

\subsection{
Translation from $\SLTL$ to $\PTLSfive$}

The translation from $\SLTL$ to $\PTLSfive$ is (slightly) more complex since  $\SLTL$ models interpret standpoint symbols, which have no natural counterpart in $\PTLSfive$ models.
As we will show, however, standpoints can be simulated in $\PTLSfive$ with fresh   propositional variables.
This, in particular, requires an additional formula (called $\aformulater_n$ below) simulating the requirements that each standpoint in an \SLTL model has at least one associated trace and that if a trace is assigned to a standpoint, it is so throughout the entire timeline.
As we show, such requirements can be easily expressed in $\PTLSfive$.

Given an $\SLTL$ formula $\aformula$, we construct a $\PTLSfive$ formula $\atranslation_2(\aformula)$, by applying  the translation map
$\atranslation_2$ such that $\atranslation_2$
is the identity map for propositional variables, it
is homomorphic for Boolean and temporal connectives, and the following hold for all $\astandpoint, \astandpoint' \in \standpoints$:
\begin{align*}
\atranslation_2(\Diamond_{*} \aformulabis) 
&\egdef 
\Diamond \atranslation_2(\aformulabis),
&
\atranslation_2(\Diamond_{\astandpoint} \aformulabis) 
&\egdef
        \Diamond (\astandpoint \wedge \atranslation_2(\aformulabis)),
\\
\atranslation_2(\Box_{*} \aformulabis) 
&\egdef 
\Box \atranslation_2(\aformulabis),
&
\atranslation_2(\Box_{\mathtt{s}} \aformulabis) 
&\egdef 
        \Box (\astandpoint \implication \atranslation_2(\aformulabis)),
\\
& & 
\atranslation_2(\mathtt{s} \preceq \astandpoint')
&\egdef 
            \Box(\astandpoint \implication \astandpoint'). 
\end{align*}
Assuming that the standpoint symbols in $\aformula$ are $\astandpoint_1, \ldots,\astandpoint_n$, we let
\[
\aformulater_n \egdef
\bigwedge_{1 \leq i \leq n} (\Diamond \astandpoint_i)
\wedge
\Box \bigwedge_{1 \leq i \leq n} (\always \mathtt{s_i} \vee \always \neg \astandpoint_i). 
\]
As we show next, checking satisfiability of $\varphi$ in $\SLTL$ reduces to checking satisfiability of $\aformulater_n \wedge \atranslation_2(\aformula)$ in $\PTLSfive$.

\begin{lemma}\label{reducefromSLTL}
  $\aformula$ is $\SLTL$-satisfiable iff
  $\aformulater_n\!\wedge\! \atranslation_2(\aformula)$ is  $\PTLSfive$-satisfiable.
\end{lemma}

\begin{proof}[Proof sketch] 
Assume that $\aformula$ is $\SLTL$-satifiable, so
  there are an \SLTL model $\amodel = \pair{\Pi}{\lambda}$ and a trace $\altlmodel \in \Pi$
  such that
  $\amodel, \altlmodel,0 \models \aformula$.
  Let us build a $\PTLSfive$ model  $\amodel' = \pair{\Nat \times W}{L}$
  such that $W \egdef \Pi$ and $
    L( (n,\altlmodel') )  \egdef \altlmodel'(n)  \cup
    \set{\astandpoint  \mid \altlmodel' \in \lambda(\astandpoint)}
    $ for all $\pair{n}{\altlmodel'} \in \Nat \times W$. 
We  show that   $\amodel', \pair{0}{\sigma} \models \aformulater_n \wedge \atranslation_2(\aformula)$.
  First, we observe that $\amodel', \pair{0}{\altlmodel} \models  \aformulater_n$.
  This follows from the definition of $\aformulater_n$ and the fact that  for all $\altlmodel' \in \Pi$ and for all $j \geq 0$,
  we have $\astandpoint_i \in L( ( j,\altlmodel' ))$ iff
  $\altlmodel' \in \lambda(\astandpoint_i)$. 
  Second, by structural induction, we can show that for all subformulae
  $\aformulabis$ of $\aformula$,
   all $\altlmodel' \in \Pi$, and all $j \in \Nat$, we get
    $\amodel, \altlmodel', j  \models \aformulabis$  iff 
    $\amodel', \pair{j}{\altlmodel'} \models \atranslation_2(\aformulabis)$.   
    Hence $\amodel', \pair{0}{\sigma} \models \aformulater_n \wedge \atranslation_2(\aformula)$.
  
  Now, assume that there is a $\PTLSfive$ model
  $\amodel = \pair{\Nat \times W}{L}$ and $\pair{0}{w} \in \Nat \times W$ such that
  $\amodel, \pair{0}{w} \models \aformulater_n \wedge \atranslation_2(\aformula)$.
  Let us build an $\SLTL$ model $\amodel' = \pair{\Pi}{\lambda}$ and a trace $\altlmodel \in
  \Pi$ such that
  $\amodel, \altlmodel,0 \models \aformula$.
  As in the proof of Lemma~\ref{reducetoSLTL},  let $\Pi$ consist  of all traces
  \[
  \altlmodel_{w'} \egdef L((0,w')), L((1,w')), L((2,w')), L((3,w')), \ldots
  \]
  with $w' \in W$.
  This time, the definition of $\lambda$ exploits the assumption that
  $\amodel, \pair{0}{w} \models \aformulater_n$.
  In particular, to define $\lambda$  we let $\altlmodel_{w'} \in \lambda(\astandpoint_i)$
  $\equivdef$ $\amodel, \pair{0}{w'} \models  \always \astandpoint_i$,
  for all $i \in \interval{1}{n}$ and $w' \in W$. 
  By the definition of $\lambda$ and the form of $\aformulater_n$, for each
  $\astandpoint_i$  we have $\lambda(\astandpoint_i) \neq \emptyset$. This guarantees that
  $\amodel$ is an \SLTL model.
  By structural induction, we can show that for all subformulae $\aformulabis$ of $\aformula$ and
  for all $\pair{j}{w'} \in \Nat \times W$, we have
    $\amodel, \pair{j}{w'} \models \atranslation_2(\aformulabis)$   iff  
    $\amodel', \altlmodel_{w'}, j \models \aformulabis$.  
  Hence, $\amodel', \altlmodel_{w}, 0 \models \aformula$ and so
  $\aformula$ is satisfiable in $\SLTL$.
\end{proof}
\noindent As the construction of 
$\aformulater_n \wedge \atranslation_2(\aformula)$ is feasible in logarithmic-space, we obtain the following corollary.
\begin{corollary}\label{inEXPS}
$\SLTL$-satisfiability is in \expspace.
\end{corollary}
\noindent Together with Corollary~\ref{EXPShard} we obtain tight complexity bounds for $\SLTL$, which is the main result of this section:
\begin{theorem} \label{theorem-expspace-completeness}
$\SLTL$-satisfiability is  \expspace-complete.
\end{theorem}

Logarithmic-space reductions between $\SLTL$ and $\PTLSfive$ emphasize how close are these formalisms, a property that
remained unnoticed so far. This allows us to establish the \expspace-completeness of $\SLTL$-satisfiability in a transparent way.
Therefore, the analysis of the properties of the tableau-style calculus for $\SLTL$ designed
by~\citet{Gigante&GomezAlvarez&Lyon23} needs, in the best case, lead to the \expspace upper bound.

\section{\pspace  Fragment of $\SLTL$}
\label{section-pspace}

Corollary~\ref{EXPShard} can be viewed as a negative
result for the usability of $\SLTL$ (it is of course  positive in terms of knowledge
about $\SLTL$ properties), but as shown below, there is some room to find interesting
fragments that include both  $\LTL$ and $\PSL$, but can be decided in polynomial-space.
Before presenting a fragment of  $\SLTL$ in which satisfiability is 
\pspace-complete, however, we provide more intuitions behind \expspace-hardness of full \SLTL.
This is helpful to discard syntactic features that lead to high complexity.

\subsection{Computationally challenging behaviour of \SLTL}
\label{section-example-expspace}

To obtain a better understanding of the \expspace-hardness of full \SLTL and of the mismatch between our result  
in Section~\ref{section-expspace-hardness} and the \pspace bound stated by~\citet[Theorem 28]{Gigante&GomezAlvarez&Lyon23}, we  present a specific \SLTL formula $\aformula_{C}$
whose satisfiability is particularly hard to decide.
This will also prove useful for constructing \pspace fragments of \SLTL, as they should disallow features
expressing the computationally challenging behaviour of $\aformula_{C}$.

Our construction of $\aformula_{C}$  exploits the following $\LTL$ formula $C_n$, encoding a binary counter 
from 0 to $2^n-1$, where the propositional variables $\avarprop_1, \dots, \avarprop_n$ correspond to consecutive bits of the
counter with $p_1$ being the most significant bit:
\begin{multline*}
(\neg \avarprop_1 \wedge \cdots \wedge \neg \avarprop_n) \wedge
 \always(\avarprop_1 \wedge \cdots \wedge \avarprop_n \implication
\mynext(\neg \avarprop_1 \wedge \cdots \wedge \neg \avarprop_n)) 
\\
\wedge \bigwedge_{1 \leq i \leq n}
\always
\Big(
(\neg \avarprop_i \wedge \bigwedge_{i < i' \leq n} \avarprop_{i'})
\implication
\big(
\bigwedge_{i < i' \leq n} ( \mynext  \neg \avarprop_{i'} )
\wedge
\mynext \avarprop_i
\\
\wedge
\bigwedge_{1 \leq i' < i} ( \avarprop_{i'} \leftrightarrow \mynext \avarprop_{i'})
\big)
\Big).
\end{multline*}
In every position $i$ of an \LTL trace $\sigma$, the propositional variables $\avarprop_1, \dots, \avarprop_n$ encode
the bits $b_1, \dots, b_n$ (with $b_1$ being the most significant bit) such that $b_j = 1$ iff $\avarprop_j \in \sigma(i)$.
We write $\sigma,i \models \mathtt{C}=m$ if the counter has  value $m$ in the position $i$, that is $b_1 \dots b_n$
represents the number $m$.
It is easy to see that if $\sigma,i \models C_n$, then
the counter has value $0$ in the position $i$ and in every next  position this value increases by 1 until the counter reaches the value $2^n-1$.
If this is the case, in the next position the value is $0$ and the process of counting (modulo $2^n$) repeats.
Hence,
$\altlmodel, i  \models \mathtt{C}=0$ and
for any $j\geq i$, we have 
$\altlmodel, j \models \mathtt{C}=m'$ where $m' \equiv j-i \; (\mathsf{mod} \; {2^n} )$.

We use $C_n$ to define 
$\aformula_{C} \egdef
\always
\big(
\Diamond_{\astandpoint} (C_n \wedge \avarprop \wedge \mynext \always \neg \avarprop)
\big)$ (of polynomial-size in $n$).
As we show next,
if $\amodel, \altlmodel, 0 \models \aformula_{C}$ for some \SLTL model $\amodel = \pair{\Pi}{\lambda}$, then $\Pi$ must contain infinitely many traces and
for every position $i >2^n$
there are exponentially many (at least  $2^n$) traces in 
$\Pi$ which are pairwise different with respect to the valuation of $p_1, \dots, p_n$ at the position $i$.
The latter, in particular,  
seems to contradict
the first paragraph of the proof of~\cite[Lemma 27]{Gigante&GomezAlvarez&Lyon23} used to argue for \pspace satisfiability of \SLTL.

\begin{proposition}\label{proposition-expspace-examples}
  The formula $\aformula_{C}$ is $\SLTL$ satisfiable. Moreover, for each \SLTL model
  $\amodel = \pair{\Pi}{\lambda}$ and $\altlmodel \in \Pi$ such that
  $\amodel, \altlmodel, 0 \models \aformula_{C}$, the following hold:
  (1) $\lambda(\astandpoint)$ is infinite,
  (2) for each  $i > 2^n$ and  $ m \in \{0, \dots ,2^n-1\}$, there is $\sigma \in \Pi$ such that 
      $
    \amodel, \sigma, i \models \mathtt{C}=m
    $, and
  (3) for all positions $i > 2^n$, there are $\altlmodel_1$, \ldots,
    $\altlmodel_{2^n}$ in $\Pi$ such that
    $
    \card{
    \set{\altlmodel_j(i) \cap \set{\avarprop_1, \ldots, \avarprop_n} \mid j \in \{1, \dots, 2^n\} }
    } = 2^n.
    $
\end{proposition}

    \begin{proof}
      Assume that $\amodel, \altlmodel, 0 \models \aformula_{C}$. 
      Firstly, let us show that $\lambda(\astandpoint)$ is necessarily infinite.
      Since $\amodel, \altlmodel, 0 \models \always(\Diamond_{\mathtt{s}} (C_n \wedge \avarprop
      \wedge \mynext \always \neg \avarprop))$, for each $i \geq 0$,
      there is a witness trace $\altlmodel^{\dag}_i \in \lambda(\astandpoint)$
      such that $\amodel, \altlmodel^{\dag}_i, i \models C_n \wedge \avarprop \wedge \mynext \always \neg \avarprop$.
      Hence, for all $i' > i$, we have
      $\amodel, \altlmodel^{\dag}_i, i' \models \neg \avarprop$ and
      $\amodel, \altlmodel^{\dag}_{i'}, i' \models \avarprop$, and so, $\altlmodel^{\dag}_i$ is distinct from $\altlmodel^{\dag}_{i'}$.
      Consequently, $\set{\altlmodel^{\dag}_0,\altlmodel^{\dag}_1,\ldots}$ is an infinite set of distinct traces, all belonging to  $\lambda(\astandpoint)$.

     To show the second item in the proposition,  let $i > 2^n$ and let $\altlmodel^{\dag}_1,\ldots,\altlmodel^{\dag}_{2^n}$ be $2^n$ witness traces
     as defined above.
     So, for all $j \in \{1, \dots, 2^n\}$, we have $\amodel, \altlmodel^{\dag}_j, j \models C_n$. 
     Assuming that $\amodel, \altlmodel^{\dag}_1, i \models \mathtt{C=}m_1$ for some
     $m_1 \in \{0, \dots,2^n-1\}$,
     for all $j \in \{1,\dots, 2^n\}$ we have $\amodel, \altlmodel^{\dag}_j, i \models \mathtt{C=}m_j$
     with $m_j \equiv m_1 + (j-1) \; (\mathsf{mod} \; {2^n} )$.
     Since $m_1, \ldots, m_{2^n}$ are pairwise distinct,  we obtain also that
     $\card{
    \set{\altlmodel_j^{\dag}(i) \cap \set{\avarprop_1, \ldots, \avarprop_n} \mid j \in \{1, \dots, 2^n\}}
     } = 2^n
     $, and so, the third item from the proposition holds.

     Let us conclude by showing that $\aformula_{C}$ is satisfiable, as stated in the proposition.
     For this we can construct a model $\amodel = \pair{\Pi}{\lambda}$ with infinitely many traces $\Pi \egdef \set{\altlmodel_0, \altlmodel_1, \ldots}$ and with
     $\lambda(\astandpoint) \egdef \Pi$.
     Each trace $\sigma_j \in \Pi$ is such that  $\amodel, \altlmodel_j, j \models
     C_n \wedge \avarprop \wedge \mynext \always \neg \avarprop$. Note that such $\sigma_j$ needs to exist and has a uniquely
     determined valuation of $\avarprop,\avarprop_1, \dots, \avarprop_n$ in all positions not smaller than $j$, whereas their
     valuation in positions smaller than $j$ is irrelevant for our construction. 
     Then, by the form of $\aformula_{C}$, we get  $\amodel, \altlmodel, 0 \models \aformula_{C}$.
\end{proof}

\noindent
Proposition~\ref{proposition-expspace-examples} specifies the types of computationally hard formulae of \SLTL.
In forthcoming Section~\ref{section-no-ltl-in-boxs}, 
we introduce a way of restricting \SLTL, in which the behaviour of $\aformula_{C}$ cannot be reconstructed,
and which yields a \pspace-complete satisfiability problem.

\subsection{Limiting the interplay between connectives}
\label{section-no-ltl-in-boxs}

In this section we show that a  \pspace upper bound for the satisfiability problem in \SLTL can be obtained by limiting the  interplay between standpoint and temporal operators.
To this end, we focus on \SLTL formulae which do not allow for \LTL connectives occurring in the scope of standpoint modal operators; in the sequel we call this fragment $\LTLPSL$. 
By way of example,
the  formula $  \Box_{\universalstandpoint} (   \always \neg \mathit{Malf}  \to \mathit{Test} )$ from Section~\ref{section-introduction-sltl}
does not belong to $\LTLPSL$.
However the following is already an $\LTLPSL$ formula:
$   \always (\Box_{\universalstandpoint}  \neg  \mathit{Malf} )  \to \Box_{\universalstandpoint}  \mathit{Test} $, which 
states that if always in future all countries
agree that  a
medical device does not malfunction, then all these countries agree that this device is safe according to testing.
It is worth observing that 
the formula $\aformula_{C}$ from Section~\ref{section-example-expspace}, illustrating computationally challenging behaviour,
does not belong to $\LTLPSL$ (because of $\mynext \always$ in the scope of $\Diamond_{\astandpoint}$).
In contrast, any formula of the form
$\always \Box_{\universalstandpoint} \aformula_1
\wedge \always (\aformula_2 \rightarrow \mynext \aformula_3)$
where $\aformula_1, \aformula_2, \aformula_3$ are \PSL formulae (see Section~\ref{section-introduction-psl})
are in $\LTLPSL$.

As we show in the remaining part of this section, we can check satisfiability of $\LTLPSL$ formulae in \pspace by adapting the
B\"uchi automaton construction developed
for \LTL~\cite{baier2008principles,demri2016temporal} and by performing  model-theoretical constructions based
on a refinement of the finite model property for $\PSL$ (see Theorem~\ref{theorem-psl}). 

Let us fix an $\LTLPSL$ formula $\aformula$ and let $\varprop(\aformula)$ be the set of propositional variables
occurring in it. 
Let $I$ (for `inequality') be the set of  subformulae of $\aformula$ of the form $\astandpoint
\preceq \astandpoint'$.
To check satisfiability of $\aformula$, our procedure considers all possible
partitions $D = (I^+,I^-)$ of $I$ into sets $I^+$ and $I^-$; intuitively  such a partition  states that
formulae in $I^+$ are true and those in $I^-$ are false. Since $\astandpoint \preceq \astandpoint'$
is a global statement (either it is satisfied in all traces and all positions of a model, or it is not satisfied in any of them), its truth value can be non-deterministically guessed from the very beginning. Moreover, such a guess is feasible in \pspace since  \npspace = \pspace~\cite{Savitch70}.
Given $D=(I^+,I^-)$, we encode the above intuition with  the  formula
\[
\aformulater_D \egdef
\always \left( \bigwedge_{(\astandpoint \preceq \astandpoint') \in I^+} ( \astandpoint \preceq \astandpoint') 
\land 
\bigwedge_{(\astandpoint \preceq \astandpoint') \in I^-}  ( \Diamond_{\astandpoint}  \ \avarprop_{\astandpoint,\astandpoint'}
\wedge \neg \Diamond_{\astandpoint'}  \ \avarprop_{\astandpoint, \astandpoint'}) \right) 
\]
where each $\avarprop_{\astandpoint, \astandpoint'}$ is a fresh propositional variable. This can be done as $\varprop$ is countably infinite.
Note that  $\aformulater_D$ encodes that each $\mathtt{s} \preceq \mathtt{s'}$ in $I^+$ holds true always in future and that each $\mathtt{s} \preceq \mathtt{s'}$ in $I^-$ is always false.
We also define
$
\aformula_D \egdef
\aformula[I^+ \mapsto \top,I^- \mapsto \bot]
\wedge \aformulater_D
$,
where $\aformula[I^+ \mapsto \top,I^- \mapsto \bot]$ is the formula obtained from $\aformula$ by replacing each
$(\astandpoint \preceq \astandpoint') \in I^+$ with $\top$ and each
$(\astandpoint \preceq \astandpoint') \in I^-$ with $\bot$. Note that
$\aformulater_D$ is an $\LTLPSL$ formula.

The purpose of the construction of $\aformula_D$, leading to Lemma~\ref{avormulad}, is to
nondeterministically choose which inequalities $\astandpoint \preceq \astandpoint'$ hold, and in particular
to enforce the non-satisfaction of some $\astandpoint \preceq \astandpoint'$  by the existence of two traces,
see the formula $\aformulater_{D}$. Notice that $\astandpoint \preceq \astandpoint'$ occurs only positively in $\aformulater_{D}$.
In forthcoming Theorem~\ref{theorem-psl}, we show how to construct \PSL models based on $\aformula_1$, defined as a
conjunction of inequalities of the form $\astandpoint \preceq \astandpoint'$. This step is therefore instrumental
to handle the inequalities $\astandpoint \preceq \astandpoint'$. 
\begin{lemma}\label{avormulad}
Satisfiability of $\aformula$ reduces to checking if there exists a partition  $D=(I^+,I^-)$ of $I$  such that  $\aformula_D$ is satisfiable.
\end{lemma}
\noindent Therefore, in what follows we will consider  an arbitrary partition $D=(I^+,I^-)$ of $I$, and  focus on checking satisfiability of $\aformula_D$.

We write $\cl(\aformula_D)$ to denote the \defstyle{closure set} of $\aformula_D$ defined as 
the smallest  set containing all the subformulae $\aformulabis$ of $\aformula_D$ as well as
$\perp$ and $\top$,
which is 
closed under taking negations (as usually, we do not allow for
double negations by
identifying $\neg \neg \aformulabis$ with $\aformulabis$)
and such that $\aformulabis \until \aformulabis' \in  \cl(\aformula_D)$ implies  $\mynext (\aformulabis \until \aformulabis') \in \cl(\aformula_D)$. We can observe that
$\card{\cl(\aformula_D)} \leq 4 \cdot \size{\aformula_D}$
and the number of formulae in the set $\cl(\aformula_D)$, whose outermost connective is a modality of the form $\Diamond_{\astandpoint}$ or
$\Box_{\astandpoint}$, is bounded by $\size{\aformula_D}$.

As in the standard automaton construction for  \LTL~\cite{baier2008principles,demri2016temporal,Esparza&Blondin23}, we call a
set $B \subseteq \cl(\aformula)$
\defstyle{maximally consistent} if it satisfies the   properties:
\begin{itemize}
\itemsep 0 cm 
  \item $\top \in B$ and $\perp \not \in B$; 
    $\aformulabis \in B$ iff $\neg \aformulabis \notin B$, for all  $\neg \aformulabis \in \cl(\aformula_D)$,
  \item 
    $\aformulabis_1 \land \aformulabis_2 \in B$ iff $\aformulabis_1 \in B$ and  $\aformulabis_2 \in B$,
    for all $\aformulabis_1 \land \aformulabis_2 \in \cl(\aformula_D)$,
  \item     $\aformulabis_1 \until  \aformulabis_2 \in B$ iff
    $\aformulabis_2 \in B$ or $\{ \aformulabis_1, \mynext (\aformulabis_1 \until  \aformulabis_2) \} \subseteq B$, for all $\aformulabis_1 \until \aformulabis_2 \in \cl(\aformula_D)$.
  \end{itemize} 
Moreover, we introduce an additional property specific to \SLTL; we say that
$B$ is \defstyle{standpoint-consistent} if the $\PSL$ formula $\bigwedge_{\aformulabis \in B \cap \PSL} \aformulabis$ (we treat here $\PSL$ as the set of all well-formed $\PSL$ formulae)
is satisfiable in standpoint logic \PSL. 
It is worth noting that
checking standpoint-consistency is decidable and, in particular, \np-complete~\cite[Corollary 3]{alvarez2022standpoint} as we have
discussed in  Section~\ref{section-introduction-sltl} (see also the forthcoming Theorem~\ref{theorem-psl}).

We call $B \subseteq \cl(\aformula_D)$
\defstyle{s-elementary} if it is maximally consistent and  standpoint-consistent.
We will use \defstyle{s-elementary} sets as states of an automaton. 

We are ready now to  define an automaton $\aautomaton_{\aformula_D}$ which we  use for checking satisfiability of $\aformula_D$.
We let $\aautomaton_{\aformula_D}$
be a generalised nondeterministic B\"uchi automaton 
$
\aautomaton_{\aformula_D} =
( \autstates, \powerset{\varprop(\aformula_D)}, \delta, \autstates_0 , \acccon),
$
whose components are as follows:
\begin{itemize}
\itemsep 0 cm 
\item  the set  of states, $\autstates$, consists of all s-elementary  sets $B \subseteq \cl(\aformula_D)$,

\item the automaton alphabet is the powerset $\powerset{\varprop(\aformula_D)}$, 

\item the set of initial states is $\autstates_0 = \set{ B \in \autstates \mid \aformula_D \in B}$,

\item the set of accepting sets, $\acccon$, is the family containing for each $\aformulabis_1 \until \aformulabis_2 \in \cl(\aformula_D)$ the set 
$
\{ B \in \autstates \mid \aformulabis_1 \until \aformulabis_2 \notin B \text{ or } \aformulabis_2 \in B \},
$

\item the transition
  relation $\delta: \autstates \times \powerset{\varprop(\aformula_D)}
  \longrightarrow \powerset{\autstates}$ is such that $\delta(B,A) = \emptyset$ if $A \neq B \cap \varprop(\aformula_D)$,
  and otherwise $\delta(B,A)$ is the set of all s-elementary $B' \subseteq \cl(\aformula_D)$ which satisfy
 the following condition: for each $\mynext \aformulabis \in \cl(\aformula_D)$, 
we have $\mynext \aformulabis \in B$ iff $\aformulabis \in B'$. 
\end{itemize}
In what follows we show that
accepting runs of $\aautomaton_{\aformula_D}$ correspond to \SLTL models of $\aformula_D$, and so,  satisfiability of $\aformula_D$ is equivalent to nonemptiness of $\aautomaton_{\aformula_D}$ language. 
In the next lemma we show the first implication of this equivalence.

\begin{lemma}\label{lemma-automaton}
If $\aformula_D$ is $\SLTL$-satisfiable,  then the language of $\aautomaton_{\aformula_D}$ is non-empty.
\end{lemma}
\begin{proof}[Proof sketch]
Assume that  $\amodel,\altlmodel,0 \models \aformula_D$, for some \SLTL model $\amodel = (\Pi,\lambda)$ and $\altlmodel \in \Pi$.
We let states $B_0,B_1, \dots$ be such that each $B_i = \{ \aformulabis \in \cl(\aformula_D) \mid \amodel,\altlmodel,i \models \aformulabis \}$
and we let input word $A_0 A_1 \dots $ be such that each 
$A_i = B_i \cap \varprop(\aformula_D)$.
As we will show, $B_0,B_1, \dots$ is a run of $\aautomaton_{\aformula_D}$
on the word $A_0 A_1 \dots $, and this run is accepting.
For the former statement, it suffices to observe that, by the definition, each $B_i$ is s-elementary and
$B_{i+1} \in \delta(B_i,A_i)$.
To prove that the run $B_0,B_1, \dots$ is accepting, we can show that for each accepting set $F \in \acccon$, there are infinitely many $i$ with $B_i \in F$.
This follows from the standard argument used in the automaton construction for
\LTL~\cite{baier2008principles,demri2016temporal}, since the accepting sets $\acccon$ in our
construction are defined in the same way as in the case of \LTL.
As $B_0,B_1, \dots$ is an accepting run of $\aautomaton_{\aformula_D}$, the language of $\aautomaton_{\aformula_D}$ is non-empty.
\end{proof}
Next, we  will show that the  converse of Lemma~\ref{lemma-automaton} (forthcoming  Lemma~\ref{lemma-automaton-converse}) holds true.
This direction is more complex and requires establishing new properties
for $\PSL$. 
Given an accepting run   $B_0, B_1, \dots$ of $\aautomaton_{\aformula_D}$,
we aim to construct an \SLTL model of $\aformula_D$.
By
standpoint-consistency of the $B_i$'s, there exists a sequence $\pair{\amodel_0}{\aprecisi_0}, 
\pair{\amodel_1}{\aprecisi_1}, \ldots$ of pointed $\PSL$ models (i.e. pairs consisting of a $\PSL$  model and
one of its precisifications) such that for all $i \geq 0$,
we have $\amodel_i, \aprecisi_i \models \bigwedge_{\aformulabis \in B_i \cap \PSL} \aformulabis$.
From such a sequence, we  aim to construct 
an $\SLTL$
model 
$\amodel = \pair{\precisis}{\lambda}$ and a trace $\altlmodel \in \precisis$ such that $\amodel, \altlmodel \models \aformula_D$.
The main challenge is to construct the set  $\precisis$  of traces in $\amodel$ from
the sets $\precisis_0, \precisis_1,\ldots$ of precisifications  in models $\amodel_0, \amodel_1, \dots$.
This task is reminiscent of known constructions
of models from partial models, for instance, as developed
in the mosaic method for temporal logics~\citep{Marx&Mikulas&Reynolds00}, see also~\cite{Wolter&Zakharyaschev98}. 
However, in our case, the construction requires exploiting  properties specific to \PSL.

To address this challenge, we  show a new model-theoretic result  
for \PSL, which is interesting on its own as, for example, it implies \np-membership of \PSL-satisfiability.
Our result states that the satisfiability of $\PSL$ formulae can be witnessed by 
\PSL models of a specific form and size, so it establishes a ``normalised small model property'' for \PSL.
We  show the result for \PSL formulae which are relevant for our construction.
In particular, for formulae of the form $\aformula_1 \wedge \aformula_2$ such that $\aformula_1$ is a conjunction of formulae $\astandpoint \preceq \astandpoint'$
and $\aformula_2$ is a \PSL formula in negation normal form with no occurrences of $\astandpoint \preceq \astandpoint'$
(see the formula $\aformula_D$ involved in Lemma~\ref{avormulad}).
Without any loss of generality, we can assume that the universal standpoint
$\universalstandpoint$ occurs in $\aformula_1$ (indeed, we can  always add to $\aformula_1$ the formula
$\universalstandpoint \preceq \universalstandpoint$, as it is equivalent to $\top$).

Our normalised small model property states
that satisfiability of $\aformula_1 \land \aformula_2$ implies existence of some  \PSL model of a specific form.
This model has
precisifications  represented by pairs $(S,j)$ such that $j$ is a natural number and $S$ 
belongs to a set $\mathbb{S}$ defined as follows (based on the form of $\aformula_1$).
We let $R$ be the reflexive and transitive closure of
  the relation $\set{\pair{\astandpoint}{\astandpoint'}  \mid
    \astandpoint \preceq \astandpoint' \ \mbox{occurs in} \ \aformula_1}
  \cup \set{\pair{\astandpoint}{\universalstandpoint} \mid
    \astandpoint \in \standpoints(\aformula_1 \land \aformula_2)}$.
Hence, if $\pair{\astandpoint}{\astandpoint'} \in R$, then $\aformula_1$ entails   $\astandpoint \preceq \astandpoint'$.
We let $R(\astandpoint)$ be the set of all $\astandpoint'$ such that $\pair{\astandpoint}{\astandpoint'} \in R$.
Then we define the set
$\mathbb{S} \egdef 
\set{
R(\astandpoint) 
\mid \astandpoint \in \standpoints(\aformula_1 \land \aformula_2)}$.
With these symbols in hand we are ready to formulate the small model property of \PSL.

\begin{theorem}
  \label{theorem-psl}
Let $\aformula_1 \land \aformula_2$ be a \PSL-formula of the form described above, which mentions
$N_1$ standpoint symbols and
$N_2$ occurrences of $\Diamond$. 
If $\aformula_1 \land \aformula_2$ is \PSL-satisfiable, then
for every natural number 
$N > N_1 + N_2 $
there is a $\PSL$ model $\amodel^{\star} = \pair{\precisis^{\star}}{V^{\star}}$ such that:
\begin{enumerate}
\item $\precisis^{\star} = \mathbb{S} \times \{1, \dots, N\}$, 
\item $\amodel^{\star}, \pair{S^{*}}{1} \models \aformula_1 \wedge \aformula_2$ with $S^{*} = R(*)$,  and
\item for each $\pair{S}{j} \in \precisis^{\star}$, it holds that
  $S = \set{\astandpoint \in \standpoints(\aformula_1 \land \aformula_2)
  \mid \pair{S}{j} \in V^{\star}(\astandpoint)} $.
\end{enumerate}
\end{theorem}
\begin{proof}[Proof sketch]
The proof employs  principles first used to  show the small model property
for $\Sfive$~\citep[Chapter 6]{Blackburn&deRijke&Venema01},
for counting logics~\citep[Chapter 1]{Mikulas95}, and for $\PSL$ itself~\citep[Section 4.4]{GomezAlvarez19}.

Given a \PSL model $\amodel^{\dag}$ and a precisification $\aprecisi^{\dag}$ such that
$\amodel^{\dag}, \aprecisi^{\dag} \models \aformula_1 \wedge \aformula_2$, we perform
four transformations to modify
$\amodel^{\dag},\aprecisi^{\dag}$
into $\amodel^{\star}, \aprecisi^{\star}$
satisfying Statements 1--3 from the theorem, where $\aprecisi^{\star}= \pair{S^{*}}{1}$. The steps are schematised below.
\begin{center}
{\tiny 
$\amodel^{\dag}, \aprecisi^{\dag}
\step{\bf Elect}
\amodel, \aprecisi
\step{\bf Select}
\amodel^f, \aprecisi^f
\step{\bf Normalise}
\amodel^{fn}, \aprecisi^{fn}
\step{\bf Populate}
\amodel^{\star}, \aprecisi^{\star}$.
}
\end{center}
Let us describe briefly the objectives of each reduction.
\begin{description}
\itemsep 0 cm
\item[(Elect)] The goal of transforming $\amodel^{\dag}, \aprecisi^{\dag}$ into
  $\amodel, \aprecisi$ is to guarantee that the standpoint symbols labelling $\aprecisi$
  are exactly those in $S^{*}$. At most one precisification is added to  $\amodel^{\dag}$
  to obtain $\amodel$.
\item[(Select)] The construction of $\amodel^f, \aprecisi^f$ from $\amodel, \aprecisi$ amounts
  to select a {\em finite} subset of precisifications from $\amodel$ to witness the satisfaction
  of all the $\Diamond$-formulae. This step is analogous to the way the
  small model property is shown for the modal logic $\Sfive$. 
\item[(Normalise)] The construction of $\amodel^{fn}, \aprecisi^{fn}$ from $\amodel^f, \aprecisi^f$
  guarantees that for any precisification $\aprecisi$ in $\amodel^{fn}$,
  the set of standpoint symbols whose valuation contains $\aprecisi$,
  belongs to $\mathbb{S}$.
  To do so, precisifications
  are  ``copied'', preserving the satisfaction of propositional variables
  but possibly updating the valuation of standpoints.
\item[(Populate)] This step to get  $\amodel^{\star}, \aprecisi^{\star}$ from
  $\amodel^{fn}, \aprecisi^{fn}$ consists in ``copying'' precisifications,
  so that the set of precisifications can be identified with $\mathbb{S} \times \interval{1}{N}$.
\end{description}
It can be shown that $\amodel^{\star}, \aprecisi^{\star}$ constructed in this way is possible and
satisfies all Statements 1--3 from the theorem. 
\end{proof}

Interestingly, as a consequence of Theorem~\ref{theorem-psl}, we obtain the \np-membership
for \PSL-satisfiability. Indeed, given a \PSL formula $\aformula$, it suffices to non-deterministically guess a partition
$D = \pair{I^+}{I^-}$ of  formulae $\astandpoint \preceq \astandpoint'$ in $\aformula$,
and to verify satisfiability of
$\bigwedge_{(\astandpoint \preceq \astandpoint') \in I^+} ( \astandpoint \preceq \astandpoint') 
\land 
\bigwedge_{(\astandpoint \preceq \astandpoint') \in I^-}  ( \Diamond_{\astandpoint}  \ \avarprop_{\astandpoint,\astandpoint'}
\wedge \neg \Diamond_{\astandpoint'}  \ \avarprop_{\astandpoint, \astandpoint'})
\wedge
\aformula[I^+ \mapsto \top,I^- \mapsto \bot]$.
The form of the obtained formula allows us to apply Theorem~\ref{theorem-psl}.
Hence, to check satisfiability,  it remains to guess a model with
$\card{\mathbb{S}} (N_1 + N_2 +1)$ precisifications. 
The whole procedure is in \np.
Above all, Theorem~\ref{theorem-psl} is a key  to show the  converse of
Lemma~\ref{lemma-automaton}, stated below.

\begin{lemma}\label{lemma-automaton-converse}
If the language of $\aautomaton_{\aformula_D}$ is not 
non-empty, 
then 
$\aformula_D$ is $\SLTL$-satisfiable
\end{lemma}
\begin{proof}[Proof sketch]
Assume that $B_0, B_1, \dots$ is an accepting run of $\aautomaton_{\aformula_D}$  on a word
$A_0 A_1  \dots$. We construct an \SLTL model $\amodel = (\Pi, \lambda)$ and a trace
$\altlmodel \in \Pi$ such that $\amodel, \altlmodel, 0 \models \aformula_D$.
For each $i \in \Nat$, by standpoint consistency of $B_i$, 
there exists a $\PSL$ model
$\amodel_i=(\precisis_i,V_i)$
and a
precisification $\aprecisi_i \in \precisis_i$ such that
$\amodel_i, \aprecisi_i \models \bigwedge_{\aformulabis \in B_i \cap \PSL} \aformulabis$.
Each formula $\bigwedge_{\aformulabis \in B_i \cap \PSL} \aformulabis$ can be written in the form
$\aformula_1 \land \aformula_2$
assumed in  Theorem~\ref{theorem-psl}.
Indeed, the only atomic formulae of the form $\astandpoint \preceq \astandpoint'$
in $\cl(\aformula_D)$ are those from $I^+$, and due to $\aformulater_D$, each
$\astandpoint \preceq \astandpoint'$ can occur only positively in $B_i$ (otherwise
$\bigwedge_{\aformulabis \in B_i \cap \PSL} \aformulabis$ is not satisfiable). 
For each $i \in \Nat$, the formula 
$\bigwedge_{\aformulabis \in B_i \cap \PSL} \aformulabis$ yields $R$, $N_1$, and $N_2$ used in Theorem~\ref{theorem-psl}.
Importantly, all these $R$ are the same, all $N_1 \leq \card{\standpoints(\aformula_D)}$, and all 
$N_2 \leq \size{\aformula_D}^2$.
Therefore we can apply 
Theorem~\ref{theorem-psl} with $N= \card{\standpoints(\aformula_D)}
+ \size{\aformula_D}^2 +1$
to obtain that
there are
$\PSL$ models
$\amodel_i' = \pair{\precisis_i'}{V_i'}$
such that
   $\precisis_i' = \mathbb{S} \times \{1, \dots ,N\}$, 
  $\amodel_i', \pair{S^{*}}{1} \models \bigwedge_{\aformulabis \in B_i \cap \PSL} \aformulabis$
  (with $S^* = R(*)$),
   and
   for all $\pair{S}{j} \in \precisis_i'$, we have
   $\set{\astandpoint \in \standpoints(\aformula_D) \mid \pair{S}{j} \in V_i'(\astandpoint)} = S$. 
%
We  use these models
to define an \SLTL model  $\amodel = \pair{\precisis}{\lambda}$ and $\altlmodel \in \precisis$.
We let $\precisis$ 
be $\set{\altlmodel_{\pair{S}{j}} \mid \pair{S}{j} \in \mathbb{S} \times \{1, \dots ,N \} }$
and $\lambda$ be  such that
\begin{itemize}
\itemsep 0 cm 
  \item for all $\astandpoint \in \standpoints(\aformula_D)$, $\atrace_{\pair{S}{j}} \in \lambda(\astandpoint)$
    iff $\astandpoint \in S$,
  \item for all $\pair{S}{j} \in \mathbb{S} \times \{1, \dots, N\}$, for all
    $k \geq 0$, we have
    $\atrace_{\pair{S}{j}}(k) \egdef
    \set{\avarprop \in \varprop(\aformula_D) \mid \pair{S}{j} \in V_k'(\avarprop)}$. 
  \end{itemize}
  Finally, we let $\altlmodel$ be $\altlmodel_{\pair{S^{*}}{1}}$.
  The construction of $\amodel$ is schematised in Figure~\ref{figure-m-construction}. 
  Observe that $\aformula_D \in B_0$.
 Hence, to prove that $\amodel, \altlmodel_{\pair{S^{*}}{1}}, 0 \models \aformula_D$,
 it suffices to show by structural induction that for all subformulae
 $\aformulabis$ of $\aformula_D$ and all $i \in \Nat$ we have
$ \aformulabis \in B_i$
iff 
$\amodel, \altlmodel_{\pair{S^{*}}{1}}, i  \models \aformulabis$.
In the basis of the induction we let  $\aformulabis$ be any subformula which does not mention \LTL connectives.
Note that for such $\aformulabis$, the equivalence follows from the construction of $\amodel'_i$.
Regarding the inductive step,
since
$\aformula$ does not mention \LTL connectives in the scope of standpoint operators,  
the cases in the inductive steps are for Boolean and \LTL connectives only.
Thus, the proof is analogous as in the case of automaton construction for
\LTL~\cite{baier2008principles,demri2016temporal}.
\end{proof}
\begin{figure}
 \begin{center}
\scalebox{0.73}{
  \begin{tikzpicture}[node distance=1.8cm,framed]
    \node (preone) {{\large $\atrace_{\pair{S_1}{1}} \in \lambda(\universalstandpoint) \cap
      \overline{\lambda(\astandpoint_1)} \cap \overline{\lambda(\astandpoint_2)}$}};
    \node (pretwo) [below=0.65cm of preone] {{\large $\atrace_{\pair{S_1}{2}} \in \lambda(\universalstandpoint) \cap
      \overline{\lambda(\astandpoint_1)} \cap \overline{\lambda(\astandpoint_2)}$}};
    \node (prethree) [below=0.65cm of pretwo] {{\large $\atrace_{\pair{S_2}{1}} \in \lambda(\universalstandpoint) \cap
      \overline{\lambda(\astandpoint_1)} \cap \lambda(\astandpoint_2)$}};
    \node (prefour) [below=0.65cm of prethree] {{\large $\atrace_{\pair{S_2}{2}} \in \lambda(\universalstandpoint) \cap
      \overline{\lambda(\astandpoint_1)} \cap \lambda(\astandpoint_2)$}};
    \node (prefive) [below=0.65cm of prefour] {{\large $\atrace_{\pair{S_3}{1}} \in \lambda(\universalstandpoint) \cap
      \lambda(\astandpoint_1) \cap \lambda(\astandpoint_2)$}};
    \node (presix) [below=0.65cm of prefive] {{\large $\atrace_{\pair{S_3}{2}} \in \lambda(\universalstandpoint) \cap
      \lambda(\astandpoint_1) \cap \lambda(\astandpoint_2)$}};

    \node[state,label=\mbox{$p,q$},rectangle] (zeroone) [right=0.3cm of preone] {$S_1,1$};
    \node () [right=0.005in of zeroone] {$\models \aformulabis_0$};
    \node[state,label=\mbox{$q$}] (zerotwo) [right=0.3cm of pretwo]  {$S_1,2$};
    \node[state,label=\mbox{$r$}] (zerothree) [right=0.3cm of prethree] {$S_2,1$};
    \node[state,label=\mbox{}] (zerofour) [right=0.3cm of prefour]  {$S_2,2$};
    \node[state,label=\mbox{$p$}] (zerofive) [right=0.3cm of prefive] {$S_3,1$};
    \node[state,label=\mbox{$q$}] (zerosix) [right=0.3cm of presix]  {$S_3,2$};

    \node[state,label=\mbox{$r,q$},ball color=yellow,rectangle] (oneone) [right of=zeroone] {$S_1,1$};
    \node () [right=0.005in of oneone] {\textcolor{black}{$\models \aformulabis_1$}};
    \node[state,label=\mbox{$q$},ball color=yellow] (onetwo) [right of=zerotwo]  {$S_1,2$};
    \node[state,label=\mbox{$r$},ball color=yellow] (onethree) [right of=zerothree] {$S_2,1$};
    \node[state,label=\mbox{$p,q,r$},ball color=yellow] (onefour) [right of=zerofour]  {$S_2,2$};
    \node[state,label=\mbox{$$},ball color=yellow] (onefive) [right of=zerofive] {$S_3,1$};
    \node[state,label=\mbox{$p,q$},ball color=yellow] (onesix) [right of=zerosix]  {$S_3,2$};

    \node[state,label=\mbox{$q$},ball color=cyan,rectangle] (twoone) [right of=oneone] {$S_1,1$};
    \node () [right=0.005in of twoone] {\textcolor{black}{$\models \aformulabis_2$}};
    \node[state,label=\mbox{$q,r$},ball color=cyan] (twotwo) [right of=onetwo]  {$S_1,2$};
    \node[state,label=\mbox{$r$},ball color=cyan] (twothree) [right of=onethree] {$S_2,1$};
    \node[state,label=\mbox{$r$},ball color=cyan] (twofour) [right of=onefour]  {$S_2,2$};
    \node[state,label=\mbox{$p,q,r$},ball color=cyan] (twofive) [right of=onefive] {$S_3,1$};
    \node[state,label=\mbox{$q$},ball color=cyan] (twosix) [right of=onesix]  {$S_3,2$};

    \node[state,label=\mbox{$p,r$},ball color=green,rectangle] (threeone) [right of=twoone] {$S_1,1$};
    \node () [right=0.005in of threeone] {\textcolor{black}{$\models \aformulabis_3$}};
    \node[state,label=\mbox{$p,q$},ball color=green] (threetwo) [right of=twotwo]  {$S_1,2$};
    \node[state,label=\mbox{$p,r$},ball color=green] (threethree) [right of=twothree] {$S_2,1$};
    \node[state,label=\mbox{$p,q$},ball color=green] (threefour) [right of=twofour]  {$S_2,2$};
    \node[state,label=\mbox{$q$},ball color=green] (threefive) [right of=twofive] {$S_3,1$};
    \node[state,label=\mbox{$p$},ball color=green] (threesix) [right of=twosix]  {$S_3,2$};


    \node () [below=0.2cm of zerosix] {$\amodel_0'$};
    \node () [below=0.2cm of onesix] {\textcolor{black}{$\amodel_1'$}};
    \node () [below=0.2cm of twosix] {\textcolor{black}{$\amodel_2'$}};
    \node () [below=0.2cm of threesix] {\textcolor{black}{$\amodel_3'$}}; 

    \path (zeroone) edge [bend left] node {} (oneone);
    \path (oneone) edge [bend left] node {} (twoone);
    \path (twoone) edge [bend left] node {} (threeone);

    \path (zerotwo) edge  node {} (onetwo);
    \path (onetwo) edge  node {} (twotwo);
    \path (twotwo) edge  node {} (threetwo);

    \path (zerothree) edge  node {} (onethree);
    \path (onethree) edge  node {} (twothree);
    \path (twothree) edge  node {} (threethree);

     \path (zerofour) edge  node {} (onefour);
    \path (onefour) edge  node {} (twofour);
    \path (twofour) edge  node {} (threefour);
    
    \path (zerofive) edge  node {} (onefive);
    \path (onefive) edge  node {} (twofive);
    \path (twofive) edge  node {} (threefive);

    \path (zerosix) edge  node {} (onesix);
    \path (onesix) edge  node {} (twosix);
    \path (twosix) edge  node {} (threesix);
    
  \end{tikzpicture}
}
\end{center}
\caption{Construction of $\amodel$ from  $\amodel_i'$'s with $N=2$, $S^* = S_1 = \set{\universalstandpoint}$,
  $S_2 = \set{\universalstandpoint, \astandpoint_2}$,
  $S_3 = \set{\universalstandpoint,\astandpoint_2, \astandpoint_1}$,
  and $\aformulabis_i = \bigwedge_{\aformulabis \in B_i \cap \PSL} \aformulabis$}
\label{figure-m-construction}
\end{figure}
Lemmas~\ref{lemma-automaton} and~\ref{lemma-automaton-converse} allow us to reduce the satisfiability  
of $\aformula_D$ to checking the nonemptiness of $\aautomaton_{\aformula_D}$ language.
This, as we show next,
leads to tight \pspace complexity bound for satisfiability checking.

\begin{theorem}
\label{theorem-ltlpsl}
\LTLPSL-satisfiability problem  is \pspace-complete. 
\end{theorem}

\begin{proof}
  The lower bound is from the fact that
  $\LTL$ is a syntactic fragment of $\LTLPSL$ and is already \pspace-hard~\citep[Theorem 4.1]{Sistla&Clarke85}. 
For the upper bound, 
assume that we want to check if $\aformula$ is satisfiable.
Our procedure starts by guessing  a partition $D$ of  subformulae of
$\aformula$ of the form $\astandpoint \preceq \astandpoint'$
  (in nondeterministic polynomial-time), and constructing  the formula $\aformula_D$.
By Lemma~\ref{avormulad}, it remains to check if $\aformula_D$ is satisfiable.
This, however, by 
Lemmas~\ref{lemma-automaton} and~\ref{lemma-automaton-converse},
reduces to checking if the language of $\aautomaton_{\aformula_D}$ is non-empty.
The size (number of states) of $\aautomaton_{\aformula_D}$ is exponentially large, but similarly as in the automata construction for \LTL,  we can use an ``on the fly'' approach to check  in \pspace non-emptiness of $\aautomaton_{\aformula_D}$ language~\cite{baier2008principles,demri2016temporal}.
The difference between our procedure and the standard one for \LTL is that we  need to check if each  constructed `on the fly'  state $B_i$ of the automaton
is standpoint-consistent.
This, however, is feasible in \np due to \np-completeness of the satisfiability problem
for $\PSL$~\cite{alvarez2022standpoint}.
\end{proof}

As a corollary of the proof, we obtain also the following result.
\begin{corollary}
If an  $\LTLPSL$ formula $\aformula$ is satisfiable,
 it is satisfied in an \SLTL model with polynomially many (in the size of $\aformula$) traces. 
\end{corollary}

\section{Concluding Remarks}
\label{section-concluding-remarks}

We studied the computational properties
of standpoint linear temporal logic $\SLTL$. First, we proved that 
its satisfiability problem  is
\expspace-complete, contrary to the \pspace bound
claimed recently in~\cite[Theorem 28]{Gigante&GomezAlvarez&Lyon23}. 
To show this result, we  designed  reductions between $\SLTL$ and
the multi-dimensional modal logic $\PTLSfive$.
Furthermore, we proposed
 a fragment of $\SLTL$ which has 
 \pspace-complete satisfiability problem, as  $\LTL$~\cite{Gigante&GomezAlvarez&Lyon23}. 
This fragment  
disallows occurrences of 
temporal connectives in the scope of standpoint connectives; 
to show that the satisfiability problem for its formulae is in \pspace,  we followed the automata-based approach~\cite{Vardi&Wolper94}
(similar to the well-known technique for \LTL)
but for which correctness requires to prove a new  model-theoretic property about \PSL
(Theorem~\ref{theorem-psl}) that we find interesting for its own sake.
In future it would be also interesting to implement practical decision procedures for \SLTL and for the  \pspace fragment we have introduced,
apart from studying its expressive power.



\begin{ack}

  We  thank the  referees for  suggestions that help
  us to improve the quality of the document. Special thanks to the referee who pointed
  us to~\cite{Barriereetal19}. 
Przemysław A Wałęga was supported by
the EPSRC projects OASIS (EP/S032347/1), ConCuR
(EP/V050869/1) and UK FIRES (EP/S019111/1), as well as SIRIUS Centre for Scalable Data Access and Samsung Research UK.
\end{ack}




\newpage
\onecolumn
\appendix


\section*{Technical Appendix}

Due of lack of space, this technical appendix is dedicated to the proofs that could not included
in the body of the paper.

\section{Proof details for Section~\ref{section-expspace-completeness}}
\label{appendix-expspace-completeness} 

\subsection{Proof details for \Cref{reducetoSLTL}}
\label{appendix-reducetoSLTL}
\begin{proof}
In the main body of the paper, to show the first implication we defined a model $\amodel'$ and 
claimed that, by 
structural induction, we can show that for all subformulae $\aformulabis$ of $\aformula$ and
  for all $\pair{n}{w'} \in \Nat \times W$, we have
    $\amodel, \pair{n}{w'} \models \aformulabis$ iff 
    $\amodel', \altlmodel_{w'}, n \models \atranslation_1(\aformulabis)$.
We considered only one case of the induction, for a subformulae of the form $\Box \aformulabis$. Now,  we will present all the remaining cases.

First, we observe that the basis of the induction, when $\aformulabis$ is a propositional variable, follows directly from the definition of $\amodel'$.
The inductive steps for subformulae of the forms $\neg \aformulabis$ and $\aformulabis_1 \land \aformulabis_2$, follow directly from the inductive assumption.
Hence, it remains to consider the cases for subformulae of the forms $\Diamond \aformulabis$, $\mynext \aformulabis$, and $\aformulabis_1 \until \aformulabis_2$.
For $\Diamond \aformulabis$
it suffices to observe that the following statements are equivalent:
  \begin{itemize}
  \itemsep 0 cm 
  \item $\amodel, \pair{n}{w'} \models \Diamond \aformulabis$
  \item $\amodel, \pair{n}{w''} \models \aformulabis$ for some
    $w'' \in W$  \hfill 
     (by definition of $\models$)
  \item $\amodel', \altlmodel_{w''}, n \models
        \atranslation_1(\aformulabis)$ for some
        $w'' \in W$
        \hfill (by  induction hypothesis)
  \item $\amodel', \altlmodel', n \models
        \atranslation_1(\aformulabis)$ for some
        $\altlmodel' \in \Pi$
        \hfill (by definition of $\Pi$)
  \item $\amodel', \altlmodel_{w'}, n \models
        \Diamond_{*} \atranslation_1(\aformulabis)$
        \hfill (by definition of $\models$)
  \item $\amodel', \altlmodel_{w'}, n \models
        \atranslation_1(\Diamond \aformulabis)$
        \hfill (by definition of $\atranslation_1$)
  \end{itemize}
For the case $\mynext \aformulabis$, it suffices to observe that the statements below  are equivalent: 
  \begin{itemize}
  \item $\amodel, \pair{n}{w'} \models \mynext \aformulabis$
  \item $\amodel, \pair{n+1}{w'} \models \aformulabis$
    \hfill (by definition of $\models$)
  \item $\amodel', \altlmodel_{w'}, n+1 \models
        \atranslation_1(\aformulabis)$
        \hfill (by  induction hypothesis)
  \item $\amodel', \altlmodel_{w'}, n \models
        \mynext \atranslation_1(\aformulabis)$
        \hfill (by definition of $\models$)
  \item $\amodel', \altlmodel_{w'}, n \models
        \atranslation_1(\mynext \aformulabis)$
        \hfill (by definition of $\atranslation_1$)  
  \end{itemize}
Finally, for the case $\aformulabis_1 \until \aformulabis_2$, it 
suffices to observe that the statements below  are equivalent: 
  \begin{itemize}
  \item $\amodel, \pair{n}{w'} \models \aformulabis_1 \until \aformulabis_2$
  \item There is $n' \geq n$ such that
    $\amodel, \pair{n'}{w'} \models \aformulabis_2$
    and for all $n \leq n'' < n'$, we have
    $\amodel, \pair{n''}{w'} \models \aformulabis_1$
       \hfill (by definition of $\models$)
  \item There is $n' \geq n$ such that
    $\amodel', \altlmodel_{w'}, n'  \models \atranslation_1(\aformulabis_2)$
    and for all $n \leq n'' < n'$, we have
    $\amodel, \altlmodel_{w'}, n''  \models \atranslation_1(\aformulabis_1)$
    \hfill (by induction hypothesis)
  \item $\amodel', \altlmodel_{w'}, n \models
         \atranslation_1(\aformulabis_1) \until \atranslation_1(\aformulabis_2)$
        \hfill (by definition of $\models$)
  \item $\amodel', \altlmodel_{w'}, n \models
        \atranslation_1(\aformulabis_1 \until \aformulabis_2)$
        \hfill (by definition of $\atranslation_1$)  
  \end{itemize}
Therefore, we obtain that $\amodel, \altlmodel_w,0 \models \atranslation_1(\aformula)$.

For the second implication in the lemma we 
claimed that we can show, by structural induction, that  for all subformulae $\aformulabis$ of $\aformula$,
   all $\altlmodel' \in \Pi$, and  all $n \in \Nat$, we have
    $\amodel, \altlmodel', n  \models \atranslation_1(\aformulabis)$  iff 
    $\amodel', \pair{n}{\altlmodel'} \models \aformulabis$.  
However, we showed only how to handle
the case with formulas $\aformulabis_1 \until \aformulabis_2$. 
Now we will cover the remaining cases.

The basis of the induction, when $\aformulabis$ is a propositional variable, follows directly from the definition of $\amodel'$ and the inductive steps for subformulae of the forms $\neg \aformulabis$ and $\aformulabis_1 \land \aformulabis_2$, follow directly from the inductive assumption.
Since the inductive step for $\aformulabis_1 \until \aformulabis_2$ is shown in the main body of the paper, it remains to consider the cases for subformulae of the forms $\Diamond \aformulabis$, $\Box \aformulabis$, and $\mynext \aformulabis$.
For $\Diamond \aformulabis$
it suffices to observe that the following statements are equivalent:
\begin{itemize}
  \item $\amodel, \altlmodel', n  \models
    \atranslation_1(\Diamond \aformulabis)$
  \item $\amodel, \altlmodel', n \models
        \Diamond_{*} \atranslation_1(\aformulabis)$
        \hfill (by definition of $\atranslation_1$)
  \item $\amodel, \altlmodel'', n \models
        \atranslation_1(\aformulabis)$ for some
        $\altlmodel'' \in \Pi$
        \hfill (by definition of $\models$)
  \item $\amodel', (n, \altlmodel'') \models
        \aformulabis$ for some
        $\altlmodel'' \in W$
        \hfill (by  induction hypothesis)
  \item $\amodel', (n,\altlmodel') \models
        \Diamond \aformulabis$ 
        \hfill (by definition of $\models$)
\end{itemize}          
For $\Box  \aformulabis$ we observe that the following statements are equivalent: 
\begin{itemize}
  \item $\amodel, \altlmodel', n  \models
    \atranslation_1(\Box \aformulabis)$
  \item $\amodel, \altlmodel', n \models
        \Box_{*} \atranslation_1(\aformulabis)$
        \hfill (by definition of $\atranslation_1$)
  \item $\amodel, \altlmodel'', n \models
        \atranslation_1(\aformulabis)$ for all
        $\altlmodel'' \in \Pi$
        \hfill (by definition of $\models$)
  \item $\amodel', (n, \altlmodel'') \models
        \aformulabis$ for all
        $\altlmodel'' \in W$
        \hfill (by  induction hypothesis)
  \item $\amodel', (n,\altlmodel') \models
        \Box \aformulabis$ 
        \hfill (by definition of $\models$)
\end{itemize}
Finally, for $\mynext \aformulabis$ the following are equivalent:
\begin{itemize}
  \item $\amodel, \altlmodel', n  \models
    \atranslation_1(\mynext \aformulabis)$
  \item $\amodel, \altlmodel', n \models
        \mynext \atranslation_1(\aformulabis)$
        \hfill (by definition of $\atranslation_1$)
  \item $\amodel, \altlmodel', n+1 \models
        \atranslation_1(\aformulabis)$ 
        \hfill (by definition of $\models$)
  \item $\amodel', (n+1, \altlmodel') \models
        \aformulabis$ 
        \hfill (by  induction hypothesis)
  \item $\amodel', (n,\altlmodel') \models
        \mynext \aformulabis$ 
        \hfill (by definition of $\models$)
\end{itemize}
Consequently, $\amodel', \pair{0}{\altlmodel} \models \aformula$.    
\end{proof}

\subsection{Proof details for \Cref{reducefromSLTL}}
\label{appendix-reducefromSLTL}
\begin{proof}
In the proof sketch, to show the first implication, we claim that by structural induction we can show that for all subformulae
  $\aformulabis$ of $\aformula$,
   all $\altlmodel' \in \Pi$, and all $j \in \Nat$, we have
    $\amodel, \altlmodel', j  \models \aformulabis$  iff 
    $\amodel', \pair{j}{\altlmodel'} \models \atranslation_2(\aformulabis)$.  
Indeed, if $\aformulabis$ is a propositional variable, then the equivalence holds by the definition of $\amodel'$.
If $\aformulabis$ is of the form $\astandpoint \preceq \astandpoint'$, then we observe that the following are equivalent:
  \begin{itemize}
  \item $\amodel, \altlmodel', j  \models \astandpoint \preceq \astandpoint'$
  
\item $\lambda(\astandpoint) \subseteq \lambda(\astandpoint')$  \hfill (by definition of $\models$)

  \item $\altlmodel'' \in \lambda(\astandpoint)$ implies
    $\altlmodel'' \in \lambda(\astandpoint')$  for all $\altlmodel''
    \in \Pi$
    \hfill (by set-theoretical reasoning)
    
  \item $\astandpoint \in L( (j,\altlmodel'') )$ implies
    $\astandpoint' \in L( (j,\altlmodel'') )$  for all $\altlmodel'' \in W$
    \hfill (by definition of $\amodel'$)
  \item $\amodel', \pair{j}{\altlmodel''} \models \astandpoint \implication \astandpoint'$
    for all $\altlmodel'' \in W$
    \hfill (by definition of $\models$)

 \item $\amodel', \pair{j}{\altlmodel'} \models
    \Box(\astandpoint \implication \astandpoint')$
    \hfill (by definition of $\atranslation_2$)

  \item $\amodel', \pair{j}{\altlmodel'} \models
         \atranslation_2(\astandpoint \preceq \astandpoint')$
  \end{itemize} 
Hence, the  basis of the induction holds. 

The inductive steps for formulae $\neg \aformulabis$ and $\aformulabis_1 \land \aformulabis_2$, hold directly by the inductive assumption.
It remains to consider inductive steps for formulae of the forms 
$\Diamond_{*} \aformulabis$, $\Box{*} \aformulabis$,
$\Diamond_{\astandpoint} \aformulabis$, and $\Box_{\astandpoint} \aformulabis$.
For formulae of the form $\Diamond_{\astandpoint} \aformulabis$
the equivalence holds because the statements below are equivalent:
  \begin{itemize}
  \item $\amodel, \altlmodel', j  \models \Diamond_{\astandpoint}
    \aformulabis$
  \item $\amodel,\altlmodel'', j  \models \aformulabis$ for some
    $\altlmodel'' \in \lambda(\astandpoint)$
    \hfill (by definition of $\models$)
  \item $\amodel', \pair{j}{\altlmodel''} \models
    \atranslation_2(\aformulabis)$ and
    $\amodel', \pair{j}{\altlmodel''} \models \astandpoint$ for some
    $\altlmodel'' \in W$
    \hfill \quad  (by  induction hypothesis and $\astandpoint \in
    L(\pair{n}{\altlmodel''})$, respectively)
  \item $\amodel', \pair{j}{\altlmodel'} \models
    \Diamond(\astandpoint \wedge \atranslation_2(\aformulabis))$
    \hfill (by definition of $\models$)
  \item  $\amodel', \pair{j}{\altlmodel'} \models
    \atranslation_2(\Diamond_{\astandpoint} (\aformulabis))$
    \hfill (by definition of $\atranslation_2$)
  \end{itemize} 
For formulae of the form $\Box_{\astandpoint} \aformulabis$
the equivalence holds since the following are equivalent:
  \begin{itemize}
  \item $\amodel, \altlmodel', j  \models \Box_{\astandpoint}
    \aformulabis$
  \item $\amodel,\altlmodel'', j  \models \aformulabis$ for all
    $\altlmodel'' \in \lambda(\astandpoint)$
    \hfill (by definition of $\models$)
  \item$\amodel', \pair{j}{\altlmodel''} \models \astandpoint$ implies  $\amodel', \pair{j}{\altlmodel''} \models
    \atranslation_2(\aformulabis)$
     for all
    $\altlmodel'' \in W$
    \hfill \quad  (by  induction hypothesis and $\astandpoint \in
    L(\pair{n}{\altlmodel''})$)
  \item $\amodel', \pair{j}{\altlmodel'} \models
    \Box(\astandpoint \to \atranslation_2(\aformulabis))$
    \hfill (by definition of $\models$)
  \item  $\amodel', \pair{j}{\altlmodel'} \models
    \atranslation_2(\Box_{\astandpoint} (\aformulabis))$
    \hfill (by definition of $\atranslation_2$)
  \end{itemize} 
The cases for $\Diamond_{*} \aformulabis$ and $\Box_{*} \aformulabis$ are analogous (and even simpler) than the above showed cases for
$\Diamond_{\astandpoint} \aformulabis$ and $\Box_{\astandpoint} \aformulabis$, respectively.
  Consequently, $\amodel', \pair{0}{\altlmodel} \models
  \aformulater_n \wedge \atranslation_2(\aformula)$.

For the second implication,
we claimed in the proof sketch that,  by structural induction, one can show that for all subformulae $\aformulabis$ of $\aformula$ and
  for all $\pair{j}{w'} \in \Nat \times W$, we have
    $\amodel, \pair{j}{w'} \models \atranslation_2(\aformulabis)$   iff  
    $\amodel', \altlmodel_{w'}, j \models \aformulabis$.  
If $\aformulabis$ is a propositional variable, then the equivalence holds by the definition of $\amodel'$.
If $\aformulabis$ is of the form $\astandpoint \preceq \astandpoint'$, then  we observe that the statements below are equivalent:
  \begin{itemize}
  \item $\amodel, \pair{j}{w'} \models
         \atranslation_2(\astandpoint \preceq \astandpoint')$
  \item $\amodel, \pair{j}{w'} \models
    \Box(\astandpoint \implication \astandpoint')$
    \hfill (by definition of $\atranslation_2$)
  \item $\amodel, \pair{j}{w''} \models \astandpoint \implication \astandpoint'$
    for all $w'' \in W$
    \hfill (by definition of $\models$)
  \item $\amodel, \pair{0}{w''} \models \always \astandpoint \implication
    \always \astandpoint'$
    for all $w'' \in W$
    \hfill (by  $\amodel, \pair{0}{w} \models \aformulater_n$)
  \item $\altlmodel_{w''} \in \lambda(\astandpoint)$ implies
    $\altlmodel_{w''} \in \lambda(\astandpoint')$  for all $w'' \in W$
    \hfill (by definition of $\lambda$)
  \item $\altlmodel' \in \lambda(\astandpoint)$ implies
    $\altlmodel' \in \lambda(\astandpoint')$  for all $\altlmodel'
    \in \Pi$
    \hfill (by definition of $\Pi$)
  \item $\lambda(\astandpoint) \subseteq \lambda(\astandpoint')$
    \hfill (by set-theoretical reasoning)
    \item $\amodel', \altlmodel_{w'}, j \models  \astandpoint \preceq \astandpoint'$ \hfill (by definition of $\models$)
  \end{itemize}  
Therefore, the  basis of the induction holds. 

The inductive steps for formulae $\neg \aformulabis$ and $\aformulabis_1 \land \aformulabis_2$, hold directly by the inductive assumption.
It remains to consider inductive steps for formulae of the forms 
$\Diamond_{*} \aformulabis$, $\Box{*} \aformulabis$,
$\Diamond_{\astandpoint} \aformulabis$, and $\Box_{\astandpoint} \aformulabis$.
For formulae of the form $\Diamond_{\astandpoint} \aformulabis$
the equivalence holds because the statements below are equivalent:
  \begin{itemize}
  \item $\amodel, \pair{j}{w'} \models    \atranslation_2(\Diamond_{\astandpoint} \aformulabis)$
  \item $\amodel, \pair{j}{w'} \models
    \Diamond(\astandpoint \wedge \atranslation_2(\aformulabis))$
 \hfill (by definition of $\atranslation_2$)
  \item $\amodel, \pair{j}{w''} \models
    \atranslation_2(\aformulabis)$ and
    $\amodel, \pair{j}{w''} \models \astandpoint$ for some
    $w'' \in W$
    \hfill (by definition of $\models$)
  \item $\amodel, \pair{j}{w''} \models
    \atranslation_2(\aformulabis)$ and
    $\amodel, \pair{0}{w''} \models \always \astandpoint$ for some
    $w'' \in W$
    \hfill (by  $\amodel, \pair{0}{w} \models \aformulater_n$)  
  \item $\amodel',\altlmodel_{w''}, j  \models \aformulabis$ for some
    $\altlmodel_{w''} \in \lambda(\astandpoint)$
    \hfill \quad  (by  induction hypothesis and  definition of $\lambda$, respectively)
  \item $\amodel', \altlmodel_{w'}, j  \models \Diamond_{\astandpoint}
    \aformulabis$
    \hfill (by definition of $\models$)
  \end{itemize}  
For formulae of the form $\Box_{\astandpoint} \aformulabis$
the equivalence holds since the following are equivalent:
  \begin{itemize}
  \item $\amodel, \pair{j}{w'} \models    \atranslation_2(\Box_{\astandpoint} \aformulabis)$
  \item $\amodel, \pair{j}{w'} \models
    \Box(\astandpoint \to \atranslation_2(\aformulabis))$
 \hfill (by definition of $\atranslation_2$)
  \item $\amodel, \pair{j}{w''} \models \astandpoint$ implies $\amodel, \pair{j}{w''} \models
    \atranslation_2(\aformulabis)$  for all 
    $w'' \in W$
    \hfill (by definition of $\models$)
  \item $\amodel, \pair{0}{w''} \models \always \astandpoint$ implies $\amodel, \pair{j}{w''} \models
    \atranslation_2(\aformulabis)$ 
     for all
    $w'' \in W$
    \hfill (by  $\amodel, \pair{0}{w} \models \aformulater_n$)  
  \item $\amodel',\altlmodel_{w''}, j  \models \aformulabis$ for all
    $\altlmodel_{w''} \in \lambda(\astandpoint)$
    \hfill \quad  (by  induction hypothesis and  definition of $\lambda$, respectively)
  \item $\amodel', \altlmodel_{w'}, j  \models \Box_{\astandpoint}
    \aformulabis$
    \hfill (by definition of $\models$)
  \end{itemize}  
The cases for $\Diamond_{*} \aformulabis$ and $\Box_{*} \aformulabis$ are analogous (and even simpler) than the above showed cases for
$\Diamond_{\astandpoint} \aformulabis$ and $\Box_{\astandpoint} \aformulabis$, respectively.
\end{proof}

\section{Proof details for Lemma~\ref{avormulad}}
\begin{proof}
Let us fix an $\LTLPSL$ formula $\aformula$ and let $I$  be the set of all subformulae of $\aformula$ of the form $\astandpoint
\preceq \astandpoint'$.
To prove the lemma it suffices to observe that the following statements are equivalent:
  \begin{itemize}
  \item $\aformula$ is satisfiable
  \item $\amodel, \altlmodel, 0 \models \aformula$ for some  $\amodel=\pair{\Pi}{\lambda}$ and  $\altlmodel \in \Pi$
   \hfill (by definition of satisfiability)
  \item $ \amodel, \altlmodel, 0 \models 
\aformula[I^+ \mapsto \top,I^- \mapsto \bot]
\wedge \left( \bigwedge_{(\astandpoint \preceq \astandpoint') \in I^+} ( \astandpoint \preceq \astandpoint') 
\land 
\bigwedge_{(\astandpoint \preceq \astandpoint') \in I^-}  ( \neg \astandpoint \preceq \astandpoint') \right) 
$ for some $\amodel=\pair{\Pi}{\lambda}$,   $\altlmodel \in \Pi$, and partition $(I^+,I^-)$ of $I$ 
\hfill (since each $\astandpoint \preceq \astandpoint' \in I$ holds in some $\altlmodel', i$ iff it holds in all $\altlmodel', i$)

\item
$ \amodel, \altlmodel, 0 \models 
\aformula[I^+ \mapsto \top,I^- \mapsto \bot]
\wedge \aformulater_D$ for some $\amodel=\pair{\Pi}{\lambda}$,   $\altlmodel \in \Pi$, and partition $D$ of $I$  \hfill (by definition of $\aformulater_D$ and  definition of an $\SLTL$ model)

\item $\amodel, \altlmodel, 0 \models \aformula_D$ for some  $\amodel=\pair{\Pi}{\lambda}$,  $\altlmodel \in \Pi$, and partition $D$ of $I$
 \hfill (by definition of $\aformula_D$)

\item $\aformula_D$ is satisfiable for some partition $D$ of $I$
   \hfill (by definition of satisfiability)
  \end{itemize}  

\end{proof}

\section{Proof details for Theorem~\ref{theorem-psl}}
\label{appendix-theorem-psl}

In Section~\ref{section-preliminaries}, a model for \PSL is defined as a pair
$\amodel = \pair{\precisis}{V}$ where $\precisis$ is a finite non-empty set of precisifications
and $V: \standpoints \cup \varprop \to \powerset{\precisis}$ (understood as a valuation).
Below, we adopt an alternative formulation (known to be equivalent when passing from valuations
to labellings) that is more practical to write a few expressions below.
We use models of the form $\pair{\precisis}{L}$ where
$L: \precisis \longrightarrow \powerset{\standpoints \cup \varprop}$ is a \defstyle{labelling}
such that
for all $\astandpoint \in \standpoints$, we have $\astandpoint \in \bigcup_{\aprecisi \in \precisis}
L(\aprecisi)$ and $\set{\aprecisi \in \precisis \mid * \in L(\aprecisi)} = \precisis$
(counterpart of $V(\astandpoint) \neq \emptyset$ for all $\astandpoint$
and $V(\universalstandpoint) = \precisis$).

\begin{proof}
Suppose that $\aformula_1 \land \aformula_2$ is satisfiable. Consequently, there is a $\PSL$ model
$\amodel^{\dag} = \pair{\precisis^{\dag}}{L^{\dag}}$ and $\aprecisi^{\dag} \in \precisis^{\dag}$ such that
\[
\amodel^{\dag}, \aprecisi^{\dag} \models \big( \bigwedge_{\astandpoint \in
 \standpoints(\aformula_1 \wedge \aformula_2)} \Diamond_{\astandpoint} \top \big)
\wedge \aformula_1 \wedge \aformula_2.
\]
The satisfaction of
$\amodel^{\dag}, \aprecisi^{\dag} \models \big( \bigwedge_{\astandpoint \in
  \standpoints(\aformula_1 \wedge \aformula_2)} \Diamond_{\astandpoint} \top \big)$
is due to the fact that for all $\astandpoint \in \standpoints$,
there is $\aprecisi_{\astandpoint} \in \precisis$
such that 
$\astandpoint \in L^{\dag}(\aprecisi_{\astandpoint})$.

We perform
four transformations to modify
$\amodel^{\dag},\aprecisi^{\dag}$
into $\amodel^{\star}, \aprecisi^{\star}$
satisfying Statements 1--3 from the theorem, where
$\aprecisi^{\star}= \pair{S^{*}}{1}$. The steps are schematised below.
\begin{center}
$\amodel^{\dag}, \aprecisi^{\dag}
\step{\bf Elect}
\amodel, \aprecisi
\step{\bf Select}
\amodel^f, \aprecisi^f
\step{\bf Normalise}
\amodel^{fn}, \aprecisi^{fn}
\step{\bf Populate}
\amodel^{\star}, \aprecisi^{\star}$.
\end{center}
Let us describe briefly the objectives of each reduction.
\begin{description}
\item[(Elect)] The goal of transforming $\amodel^{\dag}, \aprecisi^{\dag}$ into
  $\amodel, \aprecisi$ is to guarantee that the standpoint symbols labelling $\aprecisi$
  are exactly those in $S^{*}$. At most one precisification is added to  $\amodel^{\dag}$
  to obtain $\amodel$.
\item[(Select)] The construction of $\amodel^f, \aprecisi^f$ from $\amodel, \aprecisi$ amounts
  to select a {\em finite} subset of precisifications from $\amodel$ to witness the satisfaction
  of all the $\Diamond$-formulae. This step is analogous to the way the
  small model property is shown for the modal logic $\Sfive$. 
\item[(Normalise)] The construction of $\amodel^{fn}, \aprecisi^{fn}$ from $\amodel^f, \aprecisi^f$
  guarantees that for any precisification $\aprecisi$ in $\amodel^{fn}$,
  the set of standpoint symbols whose valuation contains $\aprecisi$,
  belongs to $\mathbb{S}$.
  To do so, precisifications
  are  ``copied'', preserving the satisfaction of propositional variables
  but possibly updating the valuation of standpoints.
\item[(Populate)] This step to get  $\amodel^{\star}, \aprecisi^{\star}$ from
  $\amodel^{fn}, \aprecisi^{fn}$ consists in ``copying'' precisifications,
  so that the set of precisifications can be identified with $\mathbb{S} \times \interval{1}{N}$.
\end{description}
Here are the details of the transformations as well as the formal justifications.

(Elect) We perform a preliminary step related to $S^{*}$ so that
the standpoint symbols holding at the witness precisification are precisely those in $S^{*}$
(in the process of transforming the model, $\aprecisi^{\dag}$ becomes $\aprecisi$).
In the case $L^{\dag}(\aprecisi^{\dag}) \cap \standpoints(\aformula_1 \wedge \aformula_2) \neq S^{*}$
(necessarily $S^* \subseteq L^{\dag}(\aprecisi^{\dag})$ because $\amodel^{\dag}, \aprecisi^{\dag} \models
\aformula_1$),
we build a new model $\amodel = \pair{\precisis}{L}$ from $\amodel^{\dag}$
with a new precisification $\aprecisi$
(i.e. $\precisis \egdef \precisis^{\dag} \uplus \set{\aprecisi}$) such that
  $L(\aprecisi) \egdef S^{*} \cup L^{\dag}(\aprecisi^{\dag}) \cap \varprop(\aprecisi^{\dag})$
  ($\aprecisi$ and $\aprecisi^{\dag}$ agree on the propositional variables)
  and $L$ and $L^{\dag}$ agree on $\precisis^{\dag}$. If
  $L^{\dag}(\aprecisi^{\dag}) \cap \standpoints(\aformula_1 \wedge \aformula_2) = S^{*}$,
  then $\amodel$ is equal to $\amodel^{\dag}$, and $\aprecisi$ to $\aprecisi^{\dag}$.
  For all subformulae $\aformulabis$
  of $\big( \bigwedge_{\astandpoint \in \standpoints(\aformula_1 \wedge \aformula_2)} \Diamond_{\astandpoint} \top \big)
  \wedge \aformula_1 \wedge \aformula_2$,
  for all $\aprecisi' \in \precisis^{\dag}$, we have
  $\amodel^{\dag}, \aprecisi' \models \aformulabis$ implies
  $\amodel, \aprecisi' \models \aformulabis$, and
  $\amodel^{\dag}, \aprecisi^{\dag} \models \aformulabis$ implies
  $\amodel, \aprecisi \models \aformulabis$.
  Consequently, $\amodel, \aprecisi \models
  \big( \bigwedge_{\astandpoint \in \standpoints(\aformula_1 \wedge \aformula_2)} \Diamond_{\astandpoint} \top \big) \wedge \aformula_1 \wedge \aformula_2$. Indeed, let us check carefully each conjunct.
  \begin{itemize}
  \item $\amodel, \aprecisi \models
    \big( \bigwedge_{\astandpoint \in \standpoints(\aformula_1 \wedge \aformula_2)}
    \Diamond_{\astandpoint} \top \big)$ because
    for all $\astandpoint \in \standpoints$,
there is $\aprecisi_{\astandpoint} \in \precisis^{\dag} \subseteq \precisis$
such that 
$\astandpoint \in L(\aprecisi_{\astandpoint}) =  L^{\dag}(\aprecisi_{\astandpoint})$.
\item For all $\astandpoint \preceq \astandpoint'$ in $\aformula_1$,
  for all $\aprecisi' \in \precisis^{\dag}$, we have
  $\astandpoint \in L(\aprecisi')$ implies
  $\astandpoint \in L^{\dag}(\aprecisi')$ ($L(\aprecisi') = L^{\dag}(\aprecisi')$),
  $\astandpoint' \in L^{\dag}(\aprecisi')$ (by $\amodel^{\dag}, \aprecisi^{\dag} \models \aformula_1$),
  $\astandpoint' \in L(\aprecisi')$ ($L(\aprecisi') = L^{\dag}(\aprecisi')$).
  Similarly,
  for all $\astandpoint \preceq \astandpoint'$ in $\aformula_1$,
  $\astandpoint \in L(\aprecisi)$ (for the unique $\aprecisi \in \precisis \setminus \precisis^{\dag}$
  if any) implies
  $\astandpoint \in S^*$ by definition of $L(\aprecisi)$.
  Consequently, there is a sequence of relations
  $\universalstandpoint \preceq \astandpoint_1,
  \ldots, \astandpoint_{n-1} \preceq \astandpoint_{n}$ in $\aformula_1$ with
  $\astandpoint_n = \astandpoint$.
  Hence, $\astandpoint' \in S^*$ too and therefore
  by definition of $L(\aprecisi)$, we also get
  $\astandpoint' \in L(\aprecisi)$.
  Hence, $\amodel, \aprecisi \models \aformula_1$. 
  
\item In order to show $\amodel, \aprecisi \models \aformula_2$, we need to show that
      for all subformulae $\aformulabis$
  of $\aformula_2$,
  for all $\aprecisi' \in \precisis^{\dag}$, we have
  $\amodel^{\dag}, \aprecisi' \models \aformulabis$ implies
  $\amodel, \aprecisi' \models \aformulabis$, and
  $\amodel^{\dag}, \aprecisi^{\dag} \models \aformulabis$ implies
  $\amodel, \aprecisi \models \aformulabis$.
  This can be done by structural induction. The base case with literals
  is easy as well as the cases in the induction step with the Boolean connectives
  $\vee$ and $\wedge$. By way of example,
  here are a few more cases.
  \begin{itemize}
  \item Suppose that $\amodel^{\dag}, \aprecisi' \models \Diamond_{\astandpoint} \aformulabis$.
    There is $\aprecisi'' \in \precisis^{\dag}$ such that
    $\astandpoint \in L^{\dag}(\aprecisi'')$ and
    $\amodel^{\dag}, \aprecisi'' \models  \aformulabis$.
    By definition of $L$ and by the induction hypothesis, we get
     $\astandpoint \in L(\aprecisi'')$ and
    $\amodel, \aprecisi'' \models  \aformulabis$.
    Consequently,  $\amodel, \aprecisi' \models \Diamond_{\astandpoint} \aformulabis$.

  \item Suppose that $\amodel^{\dag}, \aprecisi' \models \Box_{\astandpoint} \aformulabis$.
  For all $\aprecisi'' \in \precisis^{\dag}$ such that
    $\astandpoint \in L^{\dag}(\aprecisi'')$, we have 
  $\amodel^{\dag}, \aprecisi'' \models  \aformulabis$.
  Hence, by the induction hypothesis and by definition of $L$,
  for all $\aprecisi'' \in \precisis^{\dag}$ such that
    $\astandpoint \in L(\aprecisi'')$, we have 
  $\amodel, \aprecisi'' \models  \aformulabis$.
  {\em Ad absurdum}, suppose that
  $\astandpoint \in L(\aprecisi)$ and
  $\amodel, \aprecisi \not \models  \aformulabis$.
  By the induction hypothesis,
  $\amodel^{\dag}, \aprecisi^{\dag} \not \models  \aformulabis$.
  Since $S^* = L(\aprecisi) \cap \standpoints(\aformula_1 \wedge \aformula_2)
  \subseteq L^{\dag}(\aprecisi^{\dag})$,
  we get $\amodel^{\dag}, \aprecisi^{\dag} \not \models  \aformulabis$
  and $\astandpoint \in L^{\dag}(\aprecisi^{\dag})$, which leads to a contradiction
  with  $\amodel^{\dag}, \aprecisi' \models \Box_{\astandpoint} \aformulabis$.
  Hence, for all $\aprecisi'' \in \precisis$ such that
    $\astandpoint \in L(\aprecisi'')$, we have 
  $\amodel, \aprecisi'' \models  \aformulabis$.
    
  \end{itemize}
  \end{itemize}

(Select)  Now, we build a finite model $\amodel^f = \pair{\precisis^f}{L^f}$ from $\amodel$
  such that $\card{\precisis^f} \leq N_1 + N_2 + 1$ (and therefore $\precisis^f$ is finite, that is why we use
  the superscript `$f$') and $\precisis^f \subseteq \precisis$,
  $\aprecisi \in \precisis^f$ and
  $\amodel^f, \aprecisi \models
  \big( \bigwedge_{\astandpoint \in \standpoints(\aformula_1 \wedge \aformula_2)} \Diamond_{\astandpoint} \top \big)
  \wedge \aformula_1 \wedge \aformula_2$.
  In a way, in order to define $\precisis^f$, we identify a relevant set of witness precisifications from $\precisis$. 
  Before providing the construction, note that for any subformula $\Diamond_{\astandpoint}
  \aformulabis$ in
  $\big( \bigwedge_{\astandpoint \in \standpoints(\aformula_1 \wedge \aformula_2)} \Diamond_{\astandpoint} \top \big)
  \wedge \aformula_1 \wedge \aformula_2$, the propositions below are equivalent:
  \begin{itemize}
  \item for some $\aprecisi' \in \precisis$, we have
        $\amodel, \aprecisi' \models \Diamond_{\astandpoint}
    \aformulabis$,
  \item for all $\aprecisi' \in \precisis$, we have
        $\amodel, \aprecisi' \models \Diamond_{\astandpoint}
    \aformulabis$.
  \end{itemize}
  Indeed, $\Diamond_{\astandpoint} \aformulabis$ can be understood as a global statement about $\amodel$
  and does not depend on the precisification $\aprecisi'$ where it is evaluated.
  A similar property holds holds true for the subformulae of the form $\Box_{\astandpoint}$. 
  Let $\aset$ the set of subformulae of $\big( \bigwedge_{\astandpoint \in \standpoints(\aformula_1 \wedge \aformula_2)} \Diamond_{\astandpoint} \top \big)
  \wedge \aformula_1 \wedge \aformula_2$
  of the form
  $\Diamond_{\astandpoint} \aformulabis$ such that $\amodel, \aprecisi \models \Diamond_{\astandpoint} \aformulabis$.
  Remember that $\aformula_2$ is in negation normal form. Furthermore, the conjunct
  $\big( \bigwedge_{\astandpoint \in \standpoints(\aformula_1 \wedge \aformula_2)} \Diamond_{\astandpoint} \top \big)$ makes
  explicit the fact that each standpoint symbol $\astandpoint$ has at least one precisification. 
  For each formula $\Diamond_{\astandpoint} \aformulabis$ in $\aset$, we pick a precisification in $\amodel$,
  say $\aprecisi_{\Diamond_{\astandpoint} \aformulabis} \in \precisis$ such that
  $\astandpoint \in L(\aprecisi_{\Diamond_{\astandpoint} \aformulabis})$
  and $\amodel, \aprecisi_{\Diamond_{\astandpoint} \aformulabis} \models \aformulabis$.
  The set $\precisis^f$ is defined as the set below:
  \[
  \set{\aprecisi} \cup
  \set{\aprecisi_{\Diamond_{\astandpoint} \aformulabis} \mid \Diamond_{\astandpoint} \aformulabis \in \aset}.
  \]
      Note that $\card{\precisis^f} \leq N_1 + N_2 + 1$.
      We write $L^f$ to denote the restriction of $L$ to the set $\precisis^f$.
      One can show by structural induction that
      for all $\aprecisi' \in \precisis^f$, for all subformulae $\aformulabis$
      of $\big( \bigwedge_{\astandpoint \in \standpoints(\aformula_1 \wedge \aformula_2)} \Diamond_{\astandpoint} \top \big)
      \wedge \aformula_1 \wedge \aformula_2$, we have
      $\amodel, \aprecisi' \models \aformulabis$
      implies
      $\amodel^f, \aprecisi' \models \aformulabis$.
      Consequently, $\amodel^f, \aprecisi \models
      \big( \bigwedge_{\astandpoint \in \standpoints(\aformula_1 \wedge \aformula_2)} \Diamond_{\astandpoint} \top \big)
      \wedge \aformula_1 \wedge \aformula_2$.
      Indeed, let us check carefully each conjunct.
      \begin{itemize}
        
      \item $\amodel^f, \aprecisi \models
    \big( \bigwedge_{\astandpoint \in \standpoints(\aformula_1 \wedge \aformula_2)}
    \Diamond_{\astandpoint} \top \big)$ because
    for all $\astandpoint \in \standpoints$,
     there is $\aprecisi_{\Diamond_{\astandpoint} \top} \in \precisis^f \subseteq \precisis$
such that 
$\astandpoint \in L^f(\aprecisi_{\Diamond_{\astandpoint} \top}) =  L(\aprecisi_{\Diamond_{\astandpoint} \top})$.

\item For all $\astandpoint \preceq \astandpoint'$ in $\aformula_1$,
  for all $\aprecisi' \in \precisis^{f}$, we have
  $\astandpoint \in L^f(\aprecisi')$ implies
  $\astandpoint \in L(\aprecisi')$ ($L^f(\aprecisi') = L(\aprecisi')$),
  $\astandpoint' \in L(\aprecisi')$ (by $\amodel, \aprecisi \models \aformula_1$),
  $\astandpoint' \in L^f(\aprecisi')$ ($L^f(\aprecisi') = L(\aprecisi')$).
  
\item
In order to show $\amodel^f, \aprecisi \models \aformula_2$, we need to show that
      for all subformulae $\aformulabis$
  of $\aformula_2$,
  for all $\aprecisi' \in \precisis^{f}$, we have
  $\amodel, \aprecisi' \models \aformulabis$ implies
  $\amodel^f, \aprecisi' \models \aformulabis$. 
  This can be done by structural induction. The base case with literals
  is easy as well as the cases in the induction step with
  the Boolean connectives $\vee$ and $\wedge$. By way of example,
  here are a few more cases.
  \begin{itemize}
  \item Suppose that $\amodel, \aprecisi' \models \Diamond_{\astandpoint} \aformulabis$.
    There is $\aprecisi_{\Diamond_{\astandpoint} \aformulabis} \in \precisis$ such that
    $\astandpoint \in L(\aprecisi_{\Diamond_{\astandpoint} \aformulabis})$ and
    $\amodel, \aprecisi_{\Diamond_{\astandpoint} \aformulabis} \models  \aformulabis$.
    By definition of $L^f$ and by the induction hypothesis, we get
     $\astandpoint \in L^f(\aprecisi_{\Diamond_{\astandpoint} \aformulabis})$ and
    $\amodel^f, \aprecisi_{\Diamond_{\astandpoint} \aformulabis} \models  \aformulabis$.
    Consequently,  $\amodel^f, \aprecisi' \models \Diamond_{\astandpoint} \aformulabis$.

  \item Suppose that $\amodel, \aprecisi' \models \Box_{\astandpoint} \aformulabis$.
    For all $\aprecisi'' \in \precisis$ such that $\astandpoint \in
    L(\aprecisi'')$, we have $\amodel, \aprecisi'' \models  \aformulabis$.
    By the induction hypothesis and since $\precisis^f \subseteq \precisis$,
    for all $\aprecisi'' \in \precisis^f$ such that $\astandpoint \in
    L^f(\aprecisi'')$, we have $\amodel^f, \aprecisi'' \models  \aformulabis$.
    Consequently, $\amodel^f, \aprecisi' \models \Box_{\astandpoint} \aformulabis$.
  \end{itemize}
      \end{itemize}

      (Normalise) From $\amodel^f$, we build a larger model $\amodel^{fn} = \pair{\precisis^{fn}}{L^{fn}}$ towards
      the satisfaction of the condition (3.) (the superscript `$fn$' is intended to mean `finite and normalised').
      Given $\aprecisi' \in \precisis^f$, there is $S \in \mathbb{S}$
      such that $S \subseteq L^f(\aprecisi')$ (by satisfaction of $\aformula_1$ in $\amodel^f$).
      Indeed, we recall that $\mathbb{S} =
     \set{
     R(\astandpoint) 
      \mid \astandpoint \in \standpoints(\aformula_1 \land \aformula_2)}$
      and for all $\aprecisi' \in \precisis^f$, $L^{f}(\aprecisi') \neq \emptyset$.
      Since $\amodel^f, \aprecisi \models \aformula_1$, if
      $\astandpoint \in L^{f}(\aprecisi')$, then
      $\mathbb{S} \ni R(\astandpoint) \subseteq L^{f}(\aprecisi')$. 
      If $S =  L^f(\aprecisi') \cap \standpoints(\aformula_1 \wedge \aformula_2)$, then we keep $\aprecisi'$ in $\precisis^{fn}$ (and $\aprecisi'$
      is understood of the first copy of itself in $\amodel^{fn}$).
      Otherwise, let $S_1, \ldots, S_{\ell}$ be the set of
      all  strict subsets of  $L^f(\aprecisi') \cap \standpoints(\aformula_1 \wedge \aformula_2)$ in $\mathbb{S}$. Note that $\ell \leq N_1$.
       From  $\amodel^{f}$ to  $\amodel^{fn}$, we replace $\aprecisi'$
      by $\ell$ new copies $\aprecisi'_1, \ldots, \aprecisi'_{\ell}$
      such that all $i \in \{1, \dots, \ell \}$,
      $L^{fn}(\aprecisi'_i) \egdef S_i \cup L^f(\aprecisi') \cap \varprop(\aformula_1 \wedge \aformula_2)$. Two distinct copies $\aprecisi'_i$ and $\aprecisi'_j$ agree on the interpretation
      of the propositional variables from $\varprop(\aformula_2)$
      but may differ on the interpretation of the standpoint symbols. 
      One can show by structural induction that
      for all $\aprecisi' \in \precisis^f$, for all copies $\aprecisi'_i$, for all
      subformulae $\aformulabis$
      of $\big( \bigwedge_{\astandpoint \in \standpoints(\aformula_1 \wedge \aformula_2)} \Diamond_{\astandpoint} \top \big)
      \wedge \aformula_1 \wedge \aformula_2$, we have
      $\amodel^f, \aprecisi' \models \aformulabis$
      implies 
      $\amodel^{fn}, \aprecisi'_i \models \aformulabis$.
      This can be done by structural induction. The base case with literals 
      is easy as well as the cases in the induction step with
      the Boolean connectives $\vee$ and $\wedge$.
      By way of example,
      here are a few more cases.
  \begin{itemize}
  \item Suppose that $\amodel^f, \aprecisi' \models \Diamond_{\astandpoint} \aformulabis$.
    There is $\aprecisi'' \in \precisis^f$ such that
    $\astandpoint \in L^f(\aprecisi'')$ and
    $\amodel^f, \aprecisi'' \models  \aformulabis$.
    If $L^f(\aprecisi'')
    \cap \standpoints(\aformula_1 \wedge \aformula_2)
    \in \mathbb{S}$, then  by the induction hypothesis
    $\amodel^{fn}, \aprecisi''_1 \models  \aformulabis$
    and $L^f(\aprecisi'') = L^{fn}(\aprecisi'')$. 
    Otherwise, 
    let $S_1, \ldots, S_{\ell}$ be the set of
    all  strict subsets of  $L^f(\aprecisi'')
    \cap \standpoints(\aformula_1 \wedge \aformula_2)$ in $\mathbb{S}$.
    We have $S_1 \cup \cdots \cup S_{\ell} = L^f(\aprecisi'')
    \cap \standpoints(\aformula_1 \wedge \aformula_2)$.
    Let $i$ be such that $\astandpoint \in S_i$.
    Consequently,
    by the induction hypothesis,
    $\amodel^{fn}, \aprecisi''_i \models  \aformulabis$
    and $\astandpoint \in L^{fn}(\aprecisi''_i)$.
    Hence,  in all cases,
    $\amodel^{fn}, \aprecisi'_j \models \Diamond_{\astandpoint} \aformulabis$ for any $j$ ($j$ depending
    on $\aprecisi'$). 
    
  \item Suppose that $\amodel^f, \aprecisi' \models \Box_{\astandpoint} \aformulabis$.   
    For all $\aprecisi'' \in \precisis^f$ such that $\astandpoint \in
    L^f(\aprecisi'')$, we have $\amodel^f, \aprecisi'' \models  \aformulabis$.
    {\em Ad absurdum}, suppose that  not $\amodel^{fn}, \aprecisi^{\circ}_j \models  \aformulabis$
    for some $\aprecisi^{\circ}_j \in \precisis^{fn}$, $\aprecisi^{\circ} \in \precisis^f$
    and $\astandpoint \in L^{fn}(\aprecisi^{\circ}_j) \cap L^f(\aprecisi^{\circ})$.
    By the induction hypothesis, $\amodel^{f}, \aprecisi^{\circ} \not \models  \aformulabis$.
    Since $\astandpoint \in  L^f(\aprecisi^{\circ})$, this leads to a contradiction with
    $\amodel^f, \aprecisi' \models \Box_{\astandpoint} \aformulabis$.
 
  \end{itemize}

      Observe for all $\aprecisi' \in \precisis^{fn}$,
      we have $L^{fn}(\aprecisi') \cap \standpoints(\aformula_1 \wedge \aformula_2) \in
      \mathbb{S}$. Moreover, for all $S \in \mathbb{S}$,
      \[
      1 \leq \card{\set{\aprecisi' \in \precisis^{fn} :
        L^{fn}(\aprecisi') \cap \standpoints(\aformula_1 \wedge \aformula_2) = S}
      }
      \leq N_1 + N_2 +1 .
      \]

      (Populate) Finally, we may need to add extra copies of precisifications
      from $\amodel^{fn}$ to get the desired final model $\amodel^{\star} = \pair{\precisis^{\star}}{L^{\star}}$.
      Assuming $\aprecisi_1$ is the first copy of the initial precisification $\aprecisi$
      satisfying $\amodel, \aprecisi \models \big(
      \bigwedge_{\astandpoint \in \standpoints(\aformula_1 \wedge \aformula_2)} \Diamond_{\astandpoint} \top \big)
      \wedge \aformula_1 \wedge \aformula_2$, we  have 
      $\amodel^{fn}, \aprecisi_1 \models \big(
      \bigwedge_{\astandpoint \in \standpoints(\aformula_1 \wedge \aformula_2)} \Diamond_{\astandpoint} \top \big)
      \wedge \aformula_1 \wedge \aformula_2$ and $L^{fn}(\aprecisi_1) \cap \standpoints(\aformula_1 \wedge \aformula_2) = S^{*}$

      We recall that $N$ from the statement of Theorem~\ref{theorem-psl} is a natural number strictly greater than $N_1 + N_2$.  
      For all $S \in \mathbb{S}$, if
      \[
      N - \card{\set{\aprecisi' \in \precisis^{fn} :
        L^{fn}(\aprecisi') \cap \standpoints(\aformula_1 \wedge \aformula_2) = S}} = K > 0,
      \]
      we pick one element $\aprecisi'$ in the above-mentioned set, and we copy it $K$ times
      leading to the additional precisifications $\aprecisi^{\star}_1, \ldots, \aprecisi^{\star}_K$
      such that  for all $i \in \{1, \dots,K \}$,
      $L^{\star}(\aprecisi^{\star}_i) \egdef S \cup L^{fn}(\aprecisi') \cap \varprop(\aformula_1 \wedge \aformula_2)$.
      The set $\precisis^{\star}$ is made of the elements of $\precisis^{fn}$ plus the additional
      precisifications. One can check that for all $S \in \mathbb{S}$,
      there are exactly $N$ precisifications in
      $\set{\aprecisi' \in \precisis^{\star} :
        L^{\star}(\aprecisi') \cap \standpoints(\aformula_1 \wedge \aformula_2) = S}$
      and $\amodel^{\star}, \aprecisi_1 \models \big(
      \bigwedge_{\astandpoint \in \standpoints(\aformula_1 \wedge \aformula_2)} \Diamond_{\astandpoint} \top \big)
      \wedge \aformula_1 \wedge \aformula_2$.
      Indeed, as shown above, we can show that 
      for all $\aprecisi' \in \precisis^{fn}$, for all copies $\aprecisi^{\star}_i$, for all
      subformulae $\aformulabis$
      of $\big( \bigwedge_{\astandpoint \in \standpoints(\aformula_1 \wedge \aformula_2)} \Diamond_{\astandpoint} \top \big)
      \wedge \aformula_1 \wedge \aformula_2$, we have
      $\amodel^{fn}, \aprecisi' \models \aformulabis$
      implies 
      $\amodel^{\star}, \aprecisi^{\star}_i \models \aformulabis$.
      The cases for $\Diamond_{\astandpoint}$ and $\Box_{\astandpoint}$ are handled
      exactly as for the previous transformation. 
      
      Consequently, we can assume that
      $\precisis^{\star} = \mathbb{S} \times \{1, \dots,N\}$
      and $\amodel^{\star}, \pair{S^{*}}{1} \models \big(
      \bigwedge_{\astandpoint \in \standpoints(\aformula_1 \wedge \aformula_2)} \Diamond_{\astandpoint} \top \big)
      \wedge \aformula_1 \wedge \aformula_2$.
      Obviously, this is sufficient to satisfy the conditions~(1.) and~(2.) from Theorem~\ref{theorem-psl}.
      As far as condition~(3.) is concerned, it is satisfied first by construction of $\amodel^{fn}$ and then
      by the way copies are made in $\amodel^{\star}$. 
\end{proof}

\section{Proof details of Lemma~\ref{lemma-automaton-converse}}

\begin{proof} Assume that $B_0, B_1, \dots$ is an accepting run of
  $\aautomaton_{\aformula_D}$  on a word
$A_0 A_1  \dots$. We recall that
$
\aautomaton_{\aformula_D} =
( \autstates, \powerset{\varprop(\aformula_D)}, \delta, \autstates_0 , \acccon)$.
Consequently, for all $i$, we have $B_{i+1} \in \delta(B_i,A_i)$. 

From such a run, we construct an \SLTL model $\amodel = (\Pi, \lambda)$ and a trace
$\altlmodel \in \Pi$ such that $\amodel, \altlmodel, 0 \models \aformula_D$.
For each $i \in \Nat$, by standpoint consistency of $B_i$, 
there exists a $\PSL$ model
$\amodel_i=(\precisis_i,V_i)$
and a
precisification $\aprecisi_i \in \precisis_i$ such that
$\amodel_i, \aprecisi_i \models \bigwedge_{\aformulabis \in B_i \cap \PSL} \aformulabis$.
Each formula $\bigwedge_{\aformulabis \in B_i \cap \PSL} \aformulabis$ can be written in the form
$\aformula_1 \land \aformula_2$
assumed in  Theorem~\ref{theorem-psl}.
Indeed, the only atomic formulae of the form $\astandpoint \preceq \astandpoint'$
in $\cl(\aformula_D)$ are those from $I^+$, and due to $\aformulater_D$, each
$\astandpoint \preceq \astandpoint'$ can occur only positively in $B_i$. Otherwise
$\bigwedge_{\aformulabis \in B_i \cap \PSL} \aformulabis$ is not satisfiable.
Indeed, {\em ad absurdum}, suppose that some $\astandpoint \preceq \astandpoint'$
in $\cl(\aformula_D)$ occurs negatively in $B_i$, i.e.
$\neg (\astandpoint \preceq \astandpoint') \in B_i$.
By construction of $\cl(\aformula_D)$, $\astandpoint \preceq \astandpoint'$
belongs to $I^+$. However, using the definition of the temporal connective $\always$,
maximal consistency of the $B_j$'s implies that 
$\always \aformulabis \in B_j$ iff $\aformulabis \in B_j$ and $\mynext \always
\aformulabis \in B_j$.
Since $\aformulater_D \in B_0$ and by maximal consistency, we can conclude that for all
$j \in \Nat$, $\astandpoint \preceq \astandpoint' \in B_j$, which leads to a contradiction.

For each $i \in \Nat$, the formula 
$\bigwedge_{\aformulabis \in B_i \cap \PSL} \aformulabis$ yields $R$, $N_1$, and $N_2$ used in Theorem~\ref{theorem-psl}.
Importantly, all these $R$ are the same, all $N_1 \leq \card{\standpoints(\aformula_D)}$, and all 
$N_2 \leq \size{\aformula_D}^2$.
Therefore we can apply 
Theorem~\ref{theorem-psl} with $N= \card{\standpoints(\aformula_D)}
+ \size{\aformula_D}^2 +1$
to obtain that
there are
$\PSL$ models
$\amodel_i' = \pair{\precisis_i'}{V_i'}$
such that
   $\precisis_i' = \mathbb{S} \times \{1, \dots ,N\}$, 
  $\amodel_i', \pair{S^{*}}{1} \models \bigwedge_{\aformulabis \in B_i \cap \PSL} \aformulabis$
  (with $S^* = R(*)$),
   and
   for all $\pair{S}{j} \in \precisis_i'$, we have
   $\set{\astandpoint \in \standpoints(\aformula_D) \mid \pair{S}{j} \in V_i'(\astandpoint)} = S$. 
%
We  use these models
to define an \SLTL model  $\amodel = \pair{\precisis}{\lambda}$ and $\altlmodel \in \precisis$.

We let $\precisis$ 
be $\set{\altlmodel_{\pair{S}{j}} \mid \pair{S}{j} \in \mathbb{S} \times \{1, \dots ,N \} }$
and $\lambda$ be  such that
\begin{itemize}
\itemsep 0 cm 
  \item for all $\astandpoint \in \standpoints(\aformula_D)$, $\atrace_{\pair{S}{j}} \in \lambda(\astandpoint)$
    iff $\astandpoint \in S$,
  \item for all $\pair{S}{j} \in \mathbb{S} \times \{1, \dots, N\}$, for all
    $k \geq 0$, we have
    $\atrace_{\pair{S}{j}}(k) \egdef
    \set{\avarprop \in \varprop(\aformula_D) \mid \pair{S}{j} \in V_k'(\avarprop)}$. 
  \end{itemize}
Finally, we let $\altlmodel$ be $\altlmodel_{\pair{S^{*}}{1}}$.
The construction of $\amodel$ is schematised in Figure~\ref{figure-m-construction}.

A few interesting properties about $\amodel$ can be noticed:
$\card{\precisis} = \card{\mathbb{S}} \times N$ and
each $\altlmodel_{\pair{S}{j}}$ is built from the propositional restriction of the valuations
$V_0, V_1, \ldots$. 

Observe that $\aformula_D \in B_0$.
Hence, to prove that $\amodel, \altlmodel_{\pair{S^{*}}{1}}, 0 \models \aformula_D (=
\aformula[I^+ \mapsto \top,I^- \mapsto \bot]
\wedge \aformulater_D)$, we need to show 
by structural induction that for all subformulae
 $\aformulabis$ of $\aformula[I^+ \mapsto \top,I^- \mapsto \bot]$ and all $i \in \Nat$ we have
$ \aformulabis \in B_i$
iff 
$\amodel, \altlmodel_{\pair{S^{*}}{1}}, i  \models \aformulabis$.

In the basis of the induction we let  $\aformulabis'$ be any subformula which does not mention \LTL connectives.
Note that for such a formula $\aformulabis'$, the equivalence follows from the construction of $\amodel'_i$.
Indeed, $\aformulabis' \in B_i$ implies $\amodel_i', \pair{S^{*}}{1} \models \aformulabis'$
by Theorem~\ref{theorem-psl}, condition (2.).
Conversely, suppose that $\amodel_i', \pair{S^{*}}{1} \models \aformulabis'$.
{\em Ad absurdum}, suppose that
$\aformulabis' \not \in B_i$, i.e. $\neg \aformulabis' \in B_i$ by maximal
consistency. Consequently, $\aformulabis'$ occurs negatively in
$\bigwedge_{\aformulabis \in B_i \cap \PSL} \aformulabis$,
and by construction of $\amodel'_i$, we get
$\amodel_i', \pair{S^{*}}{1} \not \models \aformulabis'$, which leads to a contradiction.

Regarding the inductive step,
since
$\aformula$ does not mention \LTL connectives in the scope of standpoint operators,  
the cases in the inductive steps are for Boolean and \LTL connectives only.
Thus, the proof is analogous as in the case of automaton construction for
\LTL~\cite{baier2008principles,demri2016temporal}.

In order to verify that
$\amodel, \altlmodel_{\pair{S^{*}}{1}}, 0 \models \aformulater_D$,
first note that since 
$\left( \bigwedge_{(\astandpoint \preceq \astandpoint') \in I^+} ( \astandpoint \preceq \astandpoint') 
\land 
\bigwedge_{(\astandpoint \preceq \astandpoint') \in I^-}  ( \Diamond_{\astandpoint}  \ \avarprop_{\astandpoint,\astandpoint'}
\wedge \neg \Diamond_{\astandpoint'}  \ \avarprop_{\astandpoint, \astandpoint'}) \right)
\in B_i$ for all $i$, we have that such a conjunct belongs to
each $\bigwedge_{\aformulabis \in B_i \cap \PSL} \aformulabis$.
Consequently,  for all $i$, 
we have
\[
\amodel_i', \pair{S^*}{1} \models
\bigwedge_{(\astandpoint \preceq \astandpoint') \in I^+} (\astandpoint \preceq \astandpoint')
\wedge
\bigwedge_{(\astandpoint \preceq \astandpoint') \in I^-} \neg (\astandpoint \preceq \astandpoint').
\]
Let $\astandpoint \preceq \astandpoint' \in I^+$ and suppose
that $\altlmodel_{\pair{S}{j}} \in \lambda(\astandpoint)$.
By definition of $\lambda$, we have $\astandpoint \in S$.
Since $\astandpoint \preceq \astandpoint' \in I^+$, $\astandpoint' \in R(\astandpoint)$
and therefore $\astandpoint' \in S$ too.
By definition of $\lambda$, we get $\astandpoint' \in S$
and therefore $\altlmodel_{\pair{S}{j}} \in \lambda(\astandpoint')$. 
Consequently, $\amodel, \altlmodel_{\pair{S^{*}}{1}}, 0 \models
\bigwedge_{(\astandpoint \preceq \astandpoint') \in I^+} ( \astandpoint \preceq \astandpoint') $
and therefore
 $\amodel, \altlmodel_{\pair{S^{*}}{1}}, 0 \models
\always \bigwedge_{(\astandpoint \preceq \astandpoint') \in I^+} ( \astandpoint \preceq \astandpoint')$
(because the expressions of the form $\astandpoint \preceq \astandpoint'$ state global
properties about $\amodel$).
In order to verify that
$\amodel, \altlmodel_{\pair{S^{*}}{1}}, 0 \models
\always \bigwedge_{(\astandpoint \preceq \astandpoint') \in I^-}  ( \Diamond_{\astandpoint}  \ \avarprop_{\astandpoint,\astandpoint'}
\wedge \neg \Diamond_{\astandpoint'}  \ \avarprop_{\astandpoint, \astandpoint'})$,
it is sufficient to observe that
for all $i$,
$\bigwedge_{(\astandpoint \preceq \astandpoint') \in I^-}  ( \Diamond_{\astandpoint}  \ \avarprop_{\astandpoint,\astandpoint'}
\wedge \neg \Diamond_{\astandpoint'}  \ \avarprop_{\astandpoint, \astandpoint'})
\in B_i$ and
$ \amodel_i', \pair{S^*}{1} \models
\bigwedge_{(\astandpoint \preceq \astandpoint') \in I^-}  ( \Diamond_{\astandpoint}  \ \avarprop_{\astandpoint,\astandpoint'}
\wedge \neg \Diamond_{\astandpoint'}  \ \avarprop_{\astandpoint, \astandpoint'})$ by
Theorem~\ref{theorem-psl}, condition (2.).
\end{proof}

\end{document}